\documentclass[12pt]{article}
\usepackage[margin=1in]{geometry}
\usepackage{amsmath, amssymb, amsthm, bm}
\usepackage{setspace}
\usepackage{natbib}
\usepackage[
  colorlinks=true,
  citecolor=red,
  linkcolor=blue,
  urlcolor=blue,
  filecolor=violet
]{hyperref}
\usepackage{graphicx}
\usepackage[title]{appendix}
\usepackage{xcolor}
\newtheorem{remark}{Remark}
\newtheorem{assumption}{Assumption}
\newtheorem{theorem}{Theorem}
\newtheorem{lemma}{Lemma}
\usepackage{algorithm}
\usepackage{algpseudocode}
\usepackage{comment}
\usepackage{rotating}
\usepackage{booktabs}
\usepackage{graphicx}
\usepackage{subcaption}
\usepackage{footnote}
\makesavenoteenv{algorithm}
 \usepackage{xr}

\date{This version: June 7, 2026}
\begin{document}
\title{Estimator Averaging of Local Projection and VAR Impulse Responses\thanks{We thank the participants of the AMEF 2026 and the $3$rd MNB-Fudan workshop 2025 for their valuable comments. The views expressed in this paper are those of the authors and do not necessarily reflect the views of the Magyar Nemzeti Bank. Refine.ink was used to check the paper for consistency and clarity.}\\
}
\author{Chaoyi Chen\thanks{%
Magyar Nemzeti Bank (Central Bank of Hungary),  Budapest, 1054, Hungary; Email: \texttt{chenc@mnb.hu}.}\  \and Elena Pesavento\thanks{Department of Economics, Emory University, Atlanta, GA, 30322–2240, USA; Email: \texttt{epesave@emory.edu} }\ \and Bal\'azs Vonn\'ak \thanks{%
Magyar Nemzeti Bank (Central Bank of Hungary),  Budapest, 1054, Hungary; Email: \texttt{vonnakb@mnb.hu.}} } 

\maketitle
\begin{abstract}
Local projections (LP) and vector autoregressions (VAR) are the two standard tools for impulse response analysis, but they often display a finite-sample trade-off: LP is typically less biased but more volatile, while VAR is more precise but can be biased under misspecification. We propose an easy-to-implement estimator-averaging approach that combines LP and VAR at each horizon by minimizing the mean squared error of the impulse response itself, rather than in-sample fit. We derive closed-form oracle weights for this finite-sample risk problem, develop feasible AR-sieve-bootstrap procedures, and compare them against shrinkage-based Targeted Local Projections (TLP) and an $R^2$-based model-averaging benchmark. For a benchmark class of short-memory linear data generating processes in which LP and VAR are both consistent, we establish the consistency and limiting distribution of the feasible averaged estimator. Monte Carlo results show that the estimator tracks the oracle closely and compares favorably with LP, VAR, fit-based model averaging, and shrinkage-based TLP. In an empirical application revisiting \citet{BauerSwanson2023}, estimator averaging delivers stable and economically intuitive responses for yields, activity, prices, and credit spreads.


\end{abstract}

\setcounter{page}{1}\parskip0.5em \baselineskip18pt
\numberwithin{equation}{section}
\doublespacing
\section{Introduction}
\label{sec_intro}

Local projections (LP) and vector autoregressions (VAR) are the two workhorse approaches to estimating impulse response functions (IRF). In practice, however, the two methods often deliver noticeably different estimates, especially at intermediate and long horizons. This creates a natural problem for empirical researchers: if LP and VAR provide conflicting answers, how should one combine the information in the two estimators?

A useful starting point is that the distinction is not primarily about the population object being estimated. \citet{PlagborgMollerWolf2021} show that, with sufficiently rich lag structure, LPs and VARs estimate the same population IRFs. The key difference is therefore one of finite-sample behavior. As emphasized in recent work, LP and VAR exhibit a familiar finite-sample bias--variance trade-off (see, e.g., \citealp{Lietal2024}; \citealp{MontielOleaetal2025}). LPs tend to have lower bias but higher variance, especially at intermediate and long horizons, because they estimate separate horizon-specific regressions. VARs, by imposing parametric dynamic structure, typically achieve lower variance by borrowing strength across horizons, but they can incur higher bias when the DGP deviates from a finite-order VAR (\citealp{Lietal2024}).

This makes combining LP and VAR naturally appealing from a mean-squared-error perspective: an averaged estimator can potentially exploit LP's low-bias properties and VAR's low-variance properties, provided the weights adapt to horizon-specific performance. \citet{MontielOleaEtAl2026} sharpen this trade-off in a local-misspecification framework. They show that LP remains first-order robust for inference, whereas VAR can suffer first-order bias that is relevant for coverage. This highlights a practical tension: LP is attractive when robustness to misspecification is the priority, while VAR is attractive when finite-sample precision is the priority.

Existing work on combining IRF estimators has followed several related routes. One strand takes a fit-oriented model-averaging approach. For example, \citet{HounyoJung2025} propose a two-stage scheme that averages within LPs and within VARs, and then blends the two classes. Such procedures often yield smooth, interpretable IRFs and stable weights when many horizons are pooled. However, because the objective is in-sample \textit{fit}, they do not directly target the estimation risk of the structural IRF. A closely related but distinct and contemporaneous contribution is \citet{NemtyrevBoldea2026}, who develop Targeted Local Projections (TLP), a shrinkage estimator that pulls LP impulse responses toward their SVAR counterparts in order to reduce variance at the cost of some bias.\footnote{The work of \citet{NemtyrevBoldea2026} was developed independently and contemporaneously with ours. Section~\ref{subsec:TLP} discusses the relationship between TLP and our estimator-averaging approach.} Their framework is developed explicitly under a local-misspecification asymptotic setup and is complemented by bootstrap-based inference designed to improve coverage in that setting.

We take a different route. Our approach is a symmetric estimator-averaging procedure that selects weights to minimize the expected squared error of the IRF itself. Specifically, for each horizon $h$, we choose $w_h$ to minimize the population or estimated risk---variance or MSE---of the convex combination
$\widehat{\theta}_h(w_h)
= w_h\,\widehat{\theta}_{\mathrm{LP},h}+(1-w_h)\,\widehat{\theta}_{\mathrm{VAR},h}$. Our approach is symmetric between LP and VAR and directly tied to the object of interest. Rather than selecting the best-fitting model (\citealp{HounyoJung2025}) or shrinking one estimator toward the other as a baseline (\citealp{NemtyrevBoldea2026}), we ask how to combine the two estimators so as to minimize IRF estimation risk. This prioritizes point-estimation accuracy for the IRF itself, rather than in-sample fit, by choosing weights that directly reflect the LP--VAR bias--variance trade-off and by exploiting the covariance between LP and VAR to reduce sampling noise.

Relative to the existing literature, our contribution is threefold. First, in population, we derive closed-form finite-sample infeasible oracle weights that minimize the MSE of the combined estimator and make transparent how the optimal LP share varies with the horizon and with the underlying LP--VAR trade-off. Second, we develop the asymptotic theory for the AR-sieve plug-in averaged estimator under a benchmark short-memory linear DGP in which both LP and VAR are consistent for the same population impulse response. In that benchmark, the limiting risk is variance-based, and the bootstrap bias terms in Algorithm~\ref{alg:SV-weight} are a finite-sample refinement that is asymptotically negligible but empirically useful. This positions our framework as complementary to \citet{NemtyrevBoldea2026}, who study a closely related linear combination of LP and VAR estimators under a local-to-VAR drifting DGP and derive its asymptotic bias--variance trade-off in that regime. We use Monte Carlo simulations to assess finite-sample performance across correctly specified, locally misspecified, and fixed-misspecification designs, and compare our estimator with LP, VAR, TLP shrinkage, and an $R^2$-based model-averaging benchmark. The results show that the finite-sample MSE criterion and AR-sieve-bootstrap implementation perform well in the benchmark setting and remain effective under misspecification. Third, in an empirical application revisiting the high-frequency monetary policy shocks of \citet{BauerSwanson2023}, we show that estimator averaging systematically reconciles the often volatile IV-LP and very smooth IV-VAR responses: the estimated weights put more mass on LP at short horizons and on VAR at longer horizons, and the resulting IRFs for the two-year yield, industrial production, consumer prices, and the excess bond premium are reasonably smooth, lie between LP and VAR, and are arguably more economically intuitive than either estimator on its own.

The rest of the paper is organized as follows. Section~\ref{sec:estimator} presents the estimator-averaging framework, derives the oracle weights, outlines feasible implementations, and compares the approach with shrinkage and model-averaging benchmarks. Section~\ref{sec:asy_theory} presents the required assumptions and derives the large-sample properties of the estimator. Section~\ref{sec:MC} reports Monte Carlo evidence on small-sample performance and compares estimator averaging with shrinkage and model-averaging benchmarks. Section~\ref{sec:application} revisits \citet{BauerSwanson2023} using our methods and documents the empirical pattern of weights and IRFs. Section~\ref{sec:conclusion} concludes. All mathematical proofs are collected in the appendix.

\section{Estimators}
\label{sec:estimator}
In this section, we review the definitions of LP- and VAR-based IRF estimators and introduce our LP--VAR averaged estimator. We derive the optimal weights by minimizing the MSE of the combined estimator---an approach known as \textit{estimator averaging} (see, e.g., \citealp{mittelhammer2005combining}). We then discuss two related benchmarks: the TLP shrinkage estimator of \citet{NemtyrevBoldea2026} and the $R^2$-based model-averaging rule of \citet{HounyoJung2025}. The estimator-averaging framework presented here provides the core idea and can be extended to broader settings---for example, to averaging across multiple LP estimators, multiple VAR estimators, or a larger set of IRF estimators.

\subsection{LP-Based IRF}
The local projection method directly regresses the future outcome of a target variable, $y_t$, on a current impulse variable and a set of controls. Let $Y_t$ denote the $n \times 1$ vector of observed macroeconomic variables, which includes $y_t$. The horizon-$h$ LP regression is given by
\begin{equation}
y_{t+h} \;=\; c_h \;+\; \theta_h\, x_t \;+\; \gamma_h' z_t \;+\; u_{t+h},\qquad t=1,\dots,T-h,
\label{eq:LP}
\end{equation}
where $c_h$ collects deterministic terms (e.g., a constant or trends), $x_t$ is the impulse variable of interest, and $z_t$ is a vector of control variables. Typically, $z_t$ consists of lagged values of the system variables, $z_t = (Y_{t-1}', Y_{t-2}', \dots, Y_{t-p}')'$. When $x_t$ is not an observed structural shock, identifying $\theta_h$ requires additional restrictions, such as an external instrument or a recursive ordering with appropriate contemporaneous controls. The LP estimator of the structural impulse response is the OLS (or 2SLS) estimate of the scalar coefficient on the impulse variable, denoted $\widehat\theta_{\mathrm{LP},h}$. A key advantage of this approach is that it requires no dynamic parametric assumptions beyond the linear projection itself.

\subsection{VAR-Based IRF}
Alternatively, the vector autoregression (VAR) approach (see, e.g., \citealp{Sims1980}) extrapolates the impulse response by iterating on a reduced-form linear dynamic system. A standard VAR($p$) model for the full vector of observables $Y_t$ is given by
\begin{equation}
Y_t \;=\; c \;+\; \sum_{l=1}^p A_l Y_{t-l} \;+\; v_t,
\label{eq:VAR}
\end{equation}
where $c$ contains the deterministic terms, $A_l$ are the $n \times n$ matrices of autoregressive coefficients, and $v_t$ is the vector of reduced-form innovations. This system is typically estimated equation-by-equation via OLS. To recover the structural impulse responses, an identification scheme (e.g., recursive ordering via Cholesky decomposition, or external instruments/proxy SVARs) is specified to map the reduced-form innovations to the underlying structural shocks, $v_t = M_0 \epsilon_t$. 

From the fitted dynamics and the identified structural impact matrix $M_0$, the implied vector moving average (VMA) representation is computed. The VAR-based estimator, $\widehat\theta_{\mathrm{VAR},h}$, is then defined as the relevant scalar element of the $h$-step-ahead structural VMA matrix $
\widehat\theta_{\mathrm{VAR},h} \;\equiv\; g_h\bigl(\widehat A_1, \dots, \widehat A_p, \widehat M_0\bigr)$ where $g_h(\cdot)$ is the non-linear function mapping the reduced-form VAR parameters and the identification matrix to the horizon-$h$ impulse response of $y_t$.

Throughout, $\hat\theta_{\mathrm{LP},h}$ and $\hat\theta_{\mathrm{VAR},h}$ are interpreted as estimators of the same scalar structural impulse response $\theta_h$, corresponding to the relevant element of the structural Wold response path induced by the impact matrix $M_0$ and normalized to the same unit structural shock. In designs with an observed shock, this normalization is immediate. In IV or proxy-SVAR applications, both estimators use the same external instrument and the same shock normalization, so that the averaged estimator combines two estimates of the same structural object.\footnote{In our empirical application, all responses are normalized to a 25-basis-point contractionary monetary policy shock.}

\subsection{Combined Estimators}

As is well documented, LPs and VARs exhibit a finite-sample bias--variance trade-off, motivating methods that combine information from both estimators. We first present our estimator-averaging approach and then review two related but conceptually distinct benchmarks: shrinkage and model averaging. Our proposed \textit{estimator averaging} approach chooses weights to minimize the sampling risk---variance or MSE---of the combined IRF estimator. In our implementation, the horizon-specific weights directly target the MSE of the IRF itself. We then review the \textit{shrinkage} approach, represented by the TLP method of \citet{NemtyrevBoldea2026}, which shrinks LP estimates toward their VAR/SVAR counterparts using weights motivated by asymptotic MSE under local misspecification. Finally, we review \textit{model averaging}, which chooses weights based on in-sample fit rather than IRF estimation risk; we implement this using an $R^2$-based benchmark following \citet{HounyoJung2025}.\footnote{In the Monte Carlo experiments in Section~\ref{sec:MC}, we compare our plug-in MSE-optimal estimator, implemented using the AR-sieve bootstrap, with both the TLP shrinkage benchmark and the $R^2$-based model-averaging benchmark.}

\subsubsection{Estimator Averaging with MSE-Minimizing Weights}
At each horizon $h$, define the averaged IRF
\begin{equation}
\widehat\theta_h(w_h) \;=\; w_h\,\widehat\theta_{\mathrm{LP},h} \;+\; (1-w_h)\,\widehat\theta_{\mathrm{VAR},h},\qquad w_h\in[0,1],
\label{eq:avg}
\end{equation}
where $\widehat\theta_{\mathrm{LP},h}$ is the LP-based IRF estimate, $\widehat\theta_{\mathrm{VAR},h}$ is the VAR-based IRF estimate, and $w_h$ is the weight on LP at horizon $h$.

We decompose the (population) MSE of $\widehat\theta_h(w_h)$ in \eqref{eq:avg} into variance plus squared bias\footnote{The MSE criterion treats variance and squared bias symmetrically as components of quadratic loss. This corresponds to a point-estimation risk objective, consistent with the discussion in \citet{MontielOleaetal2025} that the desirability of trading bias for variance depends on the researcher's objective.}:
\begin{align}
\label{eq:msequad}
\mathrm{MSE}_h(w_h)
&=\mathbb{E}\!\left[(\widehat\theta_h(w_h)-\theta_h)^{2}\right] \nonumber\\
&= \underbrace{w_h^2 V_{L,h} + (1-w_h)^2 V_{V,h} + 2w_h(1-w_h) C_h}_{\text{variance}}
\;+\;
\underbrace{\big(w_h\,b_{L,h} + (1-w_h)\,b_{V,h}\big)^2}_{\text{bias}^2} \nonumber\\
&=w_h^2(V_{L,h}+b_{L,h}^2) + (1-w_h)^2(V_{V,h}+b_{V,h}^2) + 2w_h(1-w_h)\big(C_h+b_{L,h}b_{V,h}\big),
\end{align}
where $b_{L,h}= \mathbb{E}[\widehat\theta_{\mathrm{LP},h}]-\theta_h$, $b_{V,h}= \mathbb{E}[\widehat\theta_{\mathrm{VAR},h}]-\theta_h,$ $V_{L,h}= \mathrm{Var}(\widehat\theta_{\mathrm{LP},h}),$ $V_{V,h}=\operatorname{Var}(\hat\theta_{\mathrm{VAR},h})$, and $C_h= \mathrm{Cov}(\widehat\theta_{\mathrm{LP},h},\widehat\theta_{\mathrm{VAR},h}).$

Let $a_h=V_{L,h}+b^2_{L,h}$, $d_h=V_{V,h}+b^2_{V,h}$, and $f_h=C_h+b_{L,h}b_{V,h}$. Differentiating \eqref{eq:msequad} with respect to $w_h$ and solving the first-order condition yields the unconstrained interior candidate
\begin{align}
\widetilde w_h^\star
=
\frac{d_h-f_h}{a_h+d_h-2f_h}.
\label{eq:tildewstar}
\end{align}
Since the averaged estimator is restricted to be a convex combination, the constrained oracle weight is obtained by truncating the interior solution to the unit interval: 
\begin{align}
 w_h^\star=\min\{1,\max\{0,\widetilde w_h^\star\}\}.
 \label{eq:wstar}
 \end{align}
Therefore, when the unconstrained minimizer lies in the interior of $[0,1]$, we have $w_h^\star=\widetilde w_h^\star$, and the expression reduces to \eqref{eq:tildewstar}. Because $w_h^{\star}$ depends on unknown population quantities, we estimate it using a semiparametric AR-sieve-bootstrap, following the time-series bootstrap literature (see, e.g., \citealp{Buhlmann1997}; \citealp{KreissPaparoditis2003}; \citealp{GoncalvesKilian2004, Goncalves2007}). Specifically, we fit an AR($p$) model with $p$ selected by BIC, resample the estimated residuals nonparametrically, simulate pseudo time series, and re-estimate the relevant objects to obtain $(\widehat a_h,\widehat d_h,\widehat f_h)$. This approach avoids committing to a fixed finite-order structural DGP and instead allows the bootstrap to approximate a broader class of short-memory linear processes through a growing-order autoregressive sieve, while still delivering a feasible estimator of $w_h^{\star}$. We refer to the procedure as the AR-sieve-bootstrap plug-in estimator, or simply the \emph{plug-in estimator}. Algorithm~\ref{alg:SV-weight} summarizes the procedure.\footnote{For ease of exposition, Algorithm~\ref{alg:SV-weight} is written for a univariate model. It can easily be extended to a multivariate $\mathrm{VAR}(p_T)$ sieve.}

\begin{algorithm}[!htbp]
\caption{AR-sieve-bootstrap plug-in weight $\widehat w_h$ (implemented jointly over $h$)}
\label{alg:SV-weight}
\begin{algorithmic}[1]
\State Select $\hat p$ (e.g., by BIC over $0\leq p\leq p_{\max,T}$, where $p_{\max,T}\to\infty$ sufficiently slowly), and choose the number of bootstrap draws $B$.\footnote{The formal theory treats the sieve order as satisfying the standard AR-sieve growth conditions. We do not prove that the BIC-selected lag automatically satisfies these conditions.}
\State Fit AR($\hat p$): $y_t=\sum_{j=1}^{\hat p}\widehat\phi_j y_{t-j}+\widehat \eta_t$; set $\tilde \eta_t=\widehat \eta_t-\bar \eta$.
\State Construct the sieve pseudo-truth $\tilde\theta_h^{sieve}$ for the same normalized structural response $\theta_h$ 
(e.g., the model-implied IRF or a long-run simulation using the same identification and shock normalization).
\For{$b=1,\ldots,B$}
\State Draw $\tilde \eta_t^{\ast(b)}$ iid from the empirical distribution of $\{\tilde \eta_t\}_{t=1}^T$.
\State Simulate $y_t^{\ast(b)}=\sum_{j=1}^{\hat p}\widehat\phi_j y_{t-j}^{\ast(b)}+\tilde \eta_t^{\ast(b)}$ (discard burn-in).
\State Re-estimate on $\{y_t^{\ast(b)}\}_{t=1}^T$ to obtain $\widehat\theta_{\mathrm{LP},h}^{\ast(b)}$ and $\widehat\theta_{\mathrm{VAR},h}^{\ast(b)}$ for all $h$.
\EndFor
\For{each $h$}
\State $\bar\theta_{\mathrm{LP},h}^{\ast}=B^{-1}\sum_{b=1}^B \widehat\theta_{\mathrm{LP},h}^{\ast(b)}$, \quad
$\bar\theta_{\mathrm{VAR},h}^{\ast}=B^{-1}\sum_{b=1}^B \widehat\theta_{\mathrm{VAR},h}^{\ast(b)}$.
\State $\widehat V_{\mathrm{LP},h}^{\,SV}=B^{-1}\sum_{b=1}^B(\widehat\theta_{\mathrm{LP},h}^{\ast(b)}-\bar\theta_{\mathrm{LP},h}^{\ast})^2$,
\quad $\widehat V_{\mathrm{VAR},h}^{\,SV}=B^{-1}\sum_{b=1}^B(\widehat\theta_{\mathrm{VAR},h}^{\ast(b)}-\bar\theta_{\mathrm{VAR},h}^{\ast})^2$.
\State $\widehat C_h^{\,SV}=B^{-1}\sum_{b=1}^B(\widehat\theta_{\mathrm{LP},h}^{\ast(b)}-\bar\theta_{\mathrm{LP},h}^{\ast})(\widehat\theta_{\mathrm{VAR},h}^{\ast(b)}-\bar\theta_{\mathrm{VAR},h}^{\ast})$.
\State $\widehat b_{\mathrm{LP},h}^{\,SV}=\bar\theta_{\mathrm{LP},h}^{\ast}-\tilde\theta_h^{\,sieve}$, \quad
$\widehat b_{\mathrm{VAR},h}^{\,SV}=\bar\theta_{\mathrm{VAR},h}^{\ast}-\tilde\theta_h^{\,sieve}$.
\State $\widehat a_h=\widehat V_{\mathrm{LP},h}^{\,SV}+(\widehat b_{\mathrm{LP},h}^{\,SV})^2$, \quad
$\widehat d_h=\widehat V_{\mathrm{VAR},h}^{\,SV}+(\widehat b_{\mathrm{VAR},h}^{\,SV})^2$, \quad
$\widehat f_h=\widehat C_h^{\,SV}+\widehat b_{\mathrm{LP},h}^{\,SV}\widehat b_{\mathrm{VAR},h}^{\,SV}$.
\State $\widehat w_h=(\widehat d_h-\widehat f_h)/(\widehat a_h+\widehat d_h-2\widehat f_h)$; clip to $[0,1]$ (with safeguards if the denominator is small).
\EndFor
\end{algorithmic}
\end{algorithm}

The bootstrap bias terms in Algorithm~\ref{alg:SV-weight} are used for finite-sample implementation of the MSE criterion. They are not required for the benchmark first-order asymptotic theory in Section~\ref{sec:asy_theory}. In the benchmark setting, both LP and VAR are asymptotically centered at \(\theta_h\), so the limiting risk is variance-covariance based and the bias contributions are first-order negligible.  We therefore treat the bootstrap bias terms as a
\emph{finite-sample refinement} of a variance-based limit rather than as
objects that the first-order theory is required to validate: they capture
finite-sample features that disappear in the first-order limit but matter
empirically.   A higher-order justification of these terms would require
additional smoothness and moment conditions on the AR-sieve approximation,
which we do not impose.\footnote{A formal asymptotic justification for the bias terms used in
Algorithm~\ref{alg:SV-weight} lies outside the variance-based limit considered here; the
contemporaneous and complementary work of \citet{NemtyrevBoldea2026}
develops such a justification in a local-to-VAR drifting framework.} Accordingly, the formal theory should be interpreted as a benchmark first-order justification for the plug-in averaging procedure, while the bias-adjusted bootstrap weights used in the simulations and empirical application are finite-sample refinements evaluated through Monte Carlo evidence.

The weight choice is a finite-sample risk problem. For fixed $T$, the LP and VAR estimators may exhibit nonzero bias, so the MSE depends on both variance and bias components. In Section~\ref{sec:asy_theory}, we derive our asymptotic results as $T\to\infty$. Working under the assumption of a short-memory linear process where the LP and VAR estimators are both consistent, these bias terms become asymptotically negligible, and the first-order limit distribution is centered at $\theta_h$. Within this same section, we subsequently establish the consistency of the bootstrap estimator, $\widehat{w}$, and the limiting theory for $\widehat{\theta}(\widehat{w})$. Note that our proposed estimator can be easily extended to $K$ estimators (see Remark \ref{remark:multivariate}).

\begin{remark}\label{remark:multivariate}
The estimator-averaging framework extends directly to more than two IRF estimators. Suppose $\widehat{\boldsymbol\theta}_h=(\hat\theta_h^{(1)},\ldots,\hat\theta_h^{(K)})^\top$ collects $K$ estimators of the same scalar IRF $\theta_h$, and define
$\boldsymbol e_h=\widehat{\boldsymbol\theta}_h-\theta_h\mathbf{1}_K$,
$\mathbf{b}_h=\mathbb{E}[\boldsymbol e_h]$,
$\boldsymbol\Sigma_h=\mathrm{Var}(\boldsymbol e_h).$
For weights $\mathbf{w}_h\in\mathbb{R}^K$ satisfying $\mathbf{1}_K^\top\mathbf{w}_h=1$, consider the averaged estimator
$\tilde\theta_h(\mathbf{w}_h)=\mathbf{w}_h^\top \widehat{\boldsymbol\theta}_h.$ Its MSE is $\mathrm{MSE}[\tilde\theta_h]=\mathbf{w}_h^\top
\big(\boldsymbol\Sigma_h+\mathbf{b}_h\mathbf{b}_h^\top\big)
\mathbf{w}_h.$ Hence, minimizing MSE subject to $\mathbf{1}_K^\top\mathbf{w}_h=1$ yields the oracle weight vector
\begin{equation}
\label{eq:wstarK}
\mathbf{w}_h^\star
=
\frac{\mathbf{G}_h^{-1}\mathbf{1}_K}
{\mathbf{1}_K^\top \mathbf{G}_h^{-1}\mathbf{1}_K},
\quad
\mathbf{G}_h=\boldsymbol\Sigma_h+\mathbf{b}_h\mathbf{b}_h^\top.
\end{equation}
Thus, the same logic applies to averaging across multiple LP- and VAR-based IRF estimators. As before, $\mathbf{w}_h^\star$ depends on unknown population quantities and must be estimated in practice, for example by bootstrap methods.
\end{remark} 

 \begin{remark} \label{remark:IV}
The averaging framework is not specific to OLS. It can combine any two estimators of the same structural impulse response, including IV-based estimators. Under the usual relevance and exogeneity conditions for a valid external instrument, one may use
$\widehat\theta_h(w_h)=w_h\widehat\theta_{\mathrm{LP},h}^{IV}+(1-w_h)\widehat\theta_{\mathrm{VAR},h}^{IV}$ where the VAR component may be a proxy-SVAR or IV-VAR estimator identified with the same instrument. The MSE formula in \eqref{eq:msequad} and the bootstrap implementation apply analogously, with all bias, variance, and covariance terms computed for the corresponding IV estimators.\footnote{Our formal theory covers generic root-$T$ consistent LP and VAR impulse-response estimators. Weak, many, or invalid instruments are beyond the scope of the paper. In the empirical application, we use the Bauer--Swanson high-frequency surprise as the external instrument and apply the same averaging procedure to the resulting IV-LP and proxy-SVAR estimates.}
\end{remark}

\begin{remark} 
Following \citet{HounyoJung2025}, one can adopt a two–stage estimator-averaging scheme: (i) first, average within the LP–based estimators and within the VAR–based estimators separately; (ii) second, average the two class-specific aggregates (LP–avg and VAR–avg) using the same MSE–minimizing weighting rule. The second-stage weight leverages the complementary strengths of LP (lower bias at short horizons) and VAR (lower variance at longer horizons) to reduce the combined estimator’s risk.
\end{remark}

As a complementary implementation, we also consider empirically optimal weights that are chosen by directly minimizing the bootstrap MSE of the combined estimator, rather than by first estimating the bias, variance, and covariance components entering $w_h^{\star}$. Concretely, for each bootstrap resample we re-estimate the LP and VAR impulse responses and numerically search over weights in $[0,1]$ for the value that minimizes the bootstrap MSE of the combined estimator.

We also study a more flexible specification in which the weight depends on the discrepancy between the two estimators:
\begin{equation}
\label{eq:empirical_optimal}
    w_h^{EO}=\frac{a}{1+b\left(\frac{\widehat{\theta}_{\mathrm{LP},h}-\widehat{\theta}_{\mathrm{VAR},h}}{\widehat{\theta}_{\mathrm{LP},h}+\widehat{\theta}_{\mathrm{VAR},h}}\right)^2},
\end{equation}
where $a$ and $b$ are non-negative parameters chosen numerically. The normalization by the sum ensures scale invariance, and the case $b=0$ nests the constant-weight specification. Intuitively, larger discrepancies between LP and VAR lead this scheme to put less weight on LP, reflecting its typically higher variance. We refer to this extension as \emph{flexible estimator averaging}. In both cases, implementation relies on the same semiparametric AR-sieve-bootstrap used for the plug-in method. Since these procedures are primarily computational alternatives and do not form the basis of our asymptotic theory, we treat them as complementary rather than central to the paper's main contribution.

\subsubsection{Shrinkage with Targeted Local Projections}
\label{subsec:TLP}
In this subsection, we review TLP as a closely related benchmark shrinkage estimator. 
\citet{NemtyrevBoldea2026} derive the estimator and establish its asymptotic properties under the local-to-VAR framework of \citet{MontielOleaEtAl2026}. In that setting, the VAR/SVAR approximation is locally misspecified: the VAR/SVAR estimator may carry an asymptotic bias, while LP remains robust. TLP addresses this asymptotic bias--variance trade-off by shrinking the LP-based IRF estimator toward its VAR/SVAR counterpart, with a horizon-specific shrinkage weight selected through an Unbiased Risk Estimator (URE) for asymptotic MSE. Thus, unlike our symmetric estimator-averaging framework, TLP treats LP as the robust baseline and VAR/SVAR as the shrinkage target under local misspecification.

To be more specific, \citet{NemtyrevBoldea2026} define the TLP estimator $\widehat{\theta}^{TLP}_h$ as the solution to a penalized LP problem at each horizon:
\begin{align}
    \widehat{\theta}^{TLP}_h
    =
    \operatorname*{arg\,min}_{\theta_h\in\Theta}
    \left(\boldsymbol{Y}^{P}_{h}-\boldsymbol{X}^{P}\theta_h\right)^{\top}
    \left(\boldsymbol{Y}^{P}_{h}-\boldsymbol{X}^{P}\theta_h\right)
    +
    \lambda_{h}\left(\theta_{h}-\widehat{\theta}_{\mathrm{VAR},h}\right)^2 .
    \label{eq:TLP_penalized}
\end{align}
Here $\boldsymbol{Y}^{P}_{h}=M_Z\boldsymbol{Y}_{h}$ and $\boldsymbol{X}^{P}=M_Z\boldsymbol{X}$ are the residualized outcome and impulse variable after partialling out the controls, with $M_Z=I-\boldsymbol{Z}(\boldsymbol{Z}^{\top}\boldsymbol{Z})^{-1}\boldsymbol{Z}^{\top}$ denoting the annihilator matrix for the control matrix $\boldsymbol{Z}$. The vectors $\boldsymbol{Y}_{h}$ and $\boldsymbol{X}$ stack $y_{t+h}$ and $x_t$, respectively, over the estimation sample, while $\lambda_h>0$ controls the degree of shrinkage toward the VAR/SVAR target $\widehat{\theta}_{\mathrm{VAR},h}$.\footnote{Equivalently, in the scalar case, the penalized problem can be written as 
$\widehat{\theta}^{TLP}_h=\omega_h^{TLP}\widehat{\theta}_{\mathrm{LP},h}+(1-\omega_h^{TLP})\widehat{\theta}_{\mathrm{VAR},h}$, where 
$\omega_h^{TLP}=(\boldsymbol{X}^{P\top}\boldsymbol{X}^{P})/(\boldsymbol{X}^{P\top}\boldsymbol{X}^{P}+\lambda_h)$. Thus, larger $\lambda_h$ implies stronger shrinkage toward the VAR/SVAR target.} The TLP estimator therefore depends on the shrinkage parameter $\lambda_h$, which governs the relative weight placed on the LP estimate and the VAR/SVAR target at each horizon. \citet{NemtyrevBoldea2026} select the shrinkage parameter at each horizon by minimizing a consistent estimator of the \emph{asymptotic} MSE.

Our MSE decomposition in \eqref{eq:msequad} is closely related to the risk
criterion in \citet{NemtyrevBoldea2026}. Both papers essentially minimize the MSE of a linear combination of LP and VAR/SVAR impulse-response estimators, and both deliver weights that are in the form of a  ratio in which the numerator captures a variance/covariance gap and the denominator adds a bias-squared term.
Beyond this shared algebraic structure, however, the two papers are different both theoretically and methodologically, and we view the contributions as complementary rather than competing. \citet{NemtyrevBoldea2026} adopt the local-to-VAR DGP of \citet{MontielOleaEtAl2026} in which the misspecification of the VAR shrinks at rate $T^{-\zeta}$ for some
$\zeta>1/4$. Under this drifting sequence, the VAR carries an
$O(T^{-\zeta})$ asymptotic bias that is first-order relevant, and the
weight is determined by a bias--variance trade-off in the asymptotic limit.
We instead work under a fixed short-memory Wold DGP in which both LP and
VAR are root-$T$ consistent for the same population impulse response.
The bias--variance trade-off in our framework is therefore a
\emph{finite-sample} phenomenon, not an asymptotic one, and the limiting
oracle weight is a pure variance/covariance ratio. Given their setup, \citet{NemtyrevBoldea2026} interpret the estimator as shrinkage of LP toward VAR, with VAR playing the role of a regularization target. We instead adopt a symmetric estimator-averaging perspective, treating LP and VAR as two estimators of the same scalar IRF and choosing weights to minimize the MSE of the combined estimator. This perspective also extends naturally to more than two estimators, as discussed in Remark~\ref{remark:multivariate}. 

\subsubsection[Model Averaging with R2-Based Weights]{Model Averaging with $R^2$-Based Weights}

In this subsection, we briefly review \emph{model averaging}, focusing on achieving the best in--sample \emph{fit}. In the spirit of \citet{HounyoJung2025}, we choose the weight between LP- and VAR-based IRF estimates using their respective in-sample coefficients of determination, $R^2$.

First, we calculate the \(R^2\) for the LP regression at each horizon \(h\), denoted as
\(R^2_{\mathrm{LP},h}\). Because the LP method estimates a separate regression for each horizon, this
measure of fit is horizon-specific. Second, we calculate the \(R^2\) for the VAR model.
In the univariate case, this is the in-sample \(R^2\) from the fitted VAR equation.
In the multivariate case, we use the in-sample \(R^2\) from the reduced-form VAR equation
corresponding to the response variable whose impulse response is being averaged. Since the VAR is estimated once on the fixed sample, this measure is horizon-invariant
and is denoted by \(R^2_{\mathrm{VAR}}\). Using these measures of in-sample fit, we assign the weight to the LP estimator for a given
horizon \(h\) proportionally to its relative explanatory power. The model averaging weight is
constructed as:
\begin{equation}
\widehat w_h^{M}
= \frac{R^2_{\mathrm{LP},h}}{R^2_{\mathrm{LP},h} + R^2_{\mathrm{VAR}}}, \quad \widehat w_h^{M} \in [0,1].
\label{eq:wr2}
\end{equation}

We then obtain the model-averaged IRF at horizon \(h\) as $\widehat \theta_h^{M}
= \widehat w_h^{M} \, \widehat \theta_{\mathrm{LP},h}
+ \bigl(1 - \widehat w_h^{M}\bigr) \widehat \theta_{\mathrm{VAR},h}.$

\section{Asymptotic Results}
\label{sec:asy_theory}
Since the averaged estimator relies on a bootstrap estimate of \(w_h^\star\), this section establishes the asymptotic properties of both \(\widehat w_h\) and \(\hat\theta_h(\widehat w_h)\). We develop the asymptotic theory for $\widehat w_h$ and $\widehat\theta_h(\widehat w_h)$ by focusing on the baseline environment where both the LP and VAR estimators are consistent. By analyzing the AR-sieve-bootstrap min-MSE procedure under a short-memory linear DGP in this benchmark setting, we provide the rigorous theoretical justification that anchors our methodology. We impose standard regularity conditions below.

Although the estimators in Section~\ref{sec:estimator} are written as finite-lag LP and VAR regressions, the benchmark asymptotic theory does not assume that the true DGP is a fixed finite-order VAR. Instead, the true process is modeled as a short-memory Wold process. Equivalently, under standard invertibility and summability conditions, this process admits an infinite-order autoregressive representation that can be approximated by a growing-order AR/VAR sieve. The finite-lag LP and VAR specifications used in estimation should therefore be viewed as feasible approximations to this short-memory benchmark, with the sieve order increasing sufficiently slowly relative to the sample size.
\begin{assumption}\label{ass:Wold1}
The vector process \(\{Y_t\}\) is covariance-stationary and purely
nondeterministic, and admits the Wold representation
\[
Y_t=\sum_{j=0}^{\infty}\Psi_j\varepsilon_{t-j},
\qquad
E[\varepsilon_t\mid \mathcal F_{t-1}]=0,
\qquad
E[\varepsilon_t\varepsilon_t']=\Sigma_\varepsilon\succ 0.
\]
Assume short memory, \(\sum_{j=0}^{\infty} j\|\Psi_j\|<\infty\), and
\(E\|\varepsilon_t\|^{4+\delta}<\infty\) for some \(\delta>0\).
\end{assumption}

\begin{assumption}\label{ass:VARb2}
For any fixed $h$, there exists a (possibly singular) $2\times 2$ matrix $\Omega_h^{(2)}$ such that
\begin{align}
\sqrt{T}\Big((\widehat\theta_{\mathrm{LP},h},\widehat\theta_{\mathrm{VAR},h})'-(\theta_h,\theta_h)'\Big)
\ \Rightarrow\ \mathcal N(0,\Omega_h^{(2)}). \label{eq:jointCLT}
\end{align}
In addition, for some $\delta>0$, $\sup_T \mathbb{E}\|Z_{T,h}\|^{2+\delta}<\infty$, where $Z_{T,h}$ is defined in
Assumption~\ref{ass:VARb3}.
\end{assumption}

\begin{assumption}\label{ass:VARb3}
For any fixed $h$, define
\begin{align}
Z_{T,h}
=
\sqrt{T}\Big((\widehat\theta_{\mathrm{LP},h},\widehat\theta_{\mathrm{VAR},h})'
-(\theta_h,\theta_h)'\Big),
\label{eq:ZTh}
\end{align}
and let $(\widehat\theta_{\mathrm{LP},h}^\ast,\widehat\theta_{\mathrm{VAR},h}^\ast)'$ be generated by the AR-sieve bootstrap in Algorithm~\ref{alg:SV-weight}. Let $\mathbb{E}^\ast[\cdot]$ denote expectation conditional on the data, and define the centered bootstrap analogue
\begin{align}
Z_{T,h}^\ast
=
\sqrt{T}\Big(
(\widehat\theta_{\mathrm{LP},h}^\ast,\widehat\theta_{\mathrm{VAR},h}^\ast)'
-
\mathbb{E}^\ast[(\widehat\theta_{\mathrm{LP},h}^\ast,\widehat\theta_{\mathrm{VAR},h}^\ast)']
\Big).
\label{eq:ZThstar}
\end{align}

\smallskip
\noindent (i) 
The AR-sieve bootstrap consistently estimates the limiting second moments of the centered and scaled LP--VAR estimator pair:
\begin{align}
\mathbb{E}^\ast[Z_{T,h}^\ast Z_{T,h}^{\ast\prime}]
\xrightarrow{p}
\Omega_h^{(2)}.
\label{eq:boot_second_moment}
\end{align}

\smallskip
\noindent (ii) 
For some $\delta>0$,
$\mathbb{E}^\ast[\|Z_{T,h}^\ast\|^{2+\delta}]=O_p(1)$ and
$\sup_T \mathbb{E}[\|Z_{T,h}\|^{2+\delta}]<\infty$.

\smallskip
\noindent (iii) $B=B(T)\to\infty$ as $T\to\infty$.
\end{assumption}

To connect the finite-sample MSE weight in (\ref{eq:tildewstar}) to the root-\(T\)
asymptotic theory, it is useful to rescale the risk components. Define $a_{T,h}=V_{L,T,h}+b_{L,T,h}^2,\quad
d_{T,h}=V_{V,T,h}+b_{V,T,h}^2,\quad
f_{T,h}=C_{T,h}+b_{L,T,h}b_{V,T,h}$,
and let $A_{T,h}=T a_{T,h},\quad
D_{T,h}=T d_{T,h},\quad
F_{T,h}=T f_{T,h}$.
Since multiplying both the numerator and denominator of (\ref{eq:tildewstar}) by \(T\)
does not change the weight, we can write
\[
\widetilde w_{T,h}^{\star}
=
\frac{d_{T,h}-f_{T,h}}
     {a_{T,h}+d_{T,h}-2f_{T,h}}
=
\frac{D_{T,h}-F_{T,h}}
     {A_{T,h}+D_{T,h}-2F_{T,h}}.
\]
In the benchmark case in which both the LP and VAR estimators are
root-\(T\) consistent and asymptotically centered at \(\theta_h\), these
scaled risk components converge to the corresponding entries of
\(\Omega_h^{(2)}\).

\begin{assumption}\label{ass:rate-nondeg}
For any fixed \(h\), the scaled oracle risk components satisfy
$(A_{T,h},D_{T,h},F_{T,h}) \to (A_h,D_h,F_h)$,
where $A_h+D_h-2F_h \ge c>0$ for some constant \(c>0\). Moreover, the limiting oracle weight, $\widetilde w_h^\star =
\frac{D_h-F_h}{A_h+D_h-2F_h}$,
lies in the interior of \([0,1]\): $\widetilde w_h^\star\in[\underline w,1-\underline w]$ for some $\underline w\in(0,1/2)$. Hence the limiting clipped weight satisfies $w_h^\star=\widetilde w_h^\star$.
\end{assumption}

Assumption~\ref{ass:Wold1} assumes a short-memory Wold DGP. Under Assumptions~\ref{ass:Wold1}--\ref{ass:VARb2} and routine regularity conditions---including nonsingularity of the relevant projection matrices and, for sieve-based VAR/LP approximations, a lag order \(p_T\) that grows sufficiently slowly with \(T\)---the LP and VAR IRF estimators are consistent and jointly asymptotically normal for each fixed horizon \(h\). Moreover, AR-sieve-bootstrap approximations are valid for broad Wold-type processes for statistics whose limits depend on second-order structure. Finally, LP and VAR target the same impulse responses asymptotically when the lag length increases (\citealp{PlagborgMollerWolf2021}).

Assumption~\ref{ass:VARb2} postulates joint root-$T$ asymptotic normality for the LP and VAR IRF estimators at a fixed horizon $h$. This condition is standard for OLS-based LP and VAR estimators under Assumption~\ref{ass:Wold1} and routine rank and weak-dependence conditions; for VAR IRFs, it follows by applying the delta method to the smooth mapping from the VAR coefficients to the horizon-$h$ impulse response. Writing $\Omega_h^{(2)}=\begin{pmatrix}\Omega_{11,h}&\Omega_{12,h}\\ \Omega_{12,h}&\Omega_{22,h}\end{pmatrix}$ as in \eqref{eq:jointCLT}, the moment condition in Assumption~\ref{ass:VARb2} implies convergence of the corresponding scaled second moments: $T \operatorname{Var}(\hat\theta_{\mathrm{LP},h})\to\Omega_{11,h}$, $T \operatorname{Var}(\hat\theta_{\mathrm{VAR},h})\to\Omega_{22,h}$, and $T \operatorname{Cov}(\hat\theta_{\mathrm{LP},h},\hat\theta_{\mathrm{VAR},h})\to\Omega_{12,h}$. Hence, in the benchmark centered root-$T$ setting, the limiting risk of the averaged estimator is $e(w)'\Omega_h^{(2)}e(w)$, where $e(w)=(w,1-w)'$. The corresponding infeasible unconstrained limiting oracle weight is
\begin{align}
\widetilde w_h^\star
=
\frac{\Omega_{22,h}-\Omega_{12,h}}
{\Omega_{11,h}+\Omega_{22,h}-2\Omega_{12,h}}.
\label{eq:wstarOmega}
\end{align}
Under Assumption~\ref{ass:rate-nondeg}, this interior solution coincides with the limiting clipped oracle weight $w_h^\star$.

Assumption~\ref{ass:VARb3} is a high-level condition requiring the AR-sieve bootstrap to consistently estimate the limiting second moments of the centered and scaled LP--VAR estimator pair at a fixed horizon $h$. We state the condition in terms of centered bootstrap second moments rather than full distributional convergence, because the MSE weight and the first-order theory only require the variance and covariance components. This formulation also avoids irrelevant shifts induced by the difference between the sieve pseudo-truth and the original-sample estimates. The moment condition in part~(ii) justifies convergence of the quadratic functionals used to construct the MSE weight. The AR-sieve bootstrap is designed for short-memory linear processes and may be less reliable under strong nonlinearity, structural breaks, unit roots or near-unit roots, long memory, heavy tails, or pronounced conditional heteroskedasticity; in the latter case, a wild or heteroskedasticity-robust bootstrap variant may be more appropriate. Assumptions~\ref{ass:VARb2} and~\ref{ass:VARb3} are stated as high-level conditions. They can be derived under standard primitive regularity conditions for stable finite-order VARs and, more generally, for short-memory linear processes approximated by a suitably growing AR sieve. A formal treatment of data-dependent lag selection rules, such as BIC or AIC, is beyond the scope of the paper. We therefore treat these lag choices as part of the feasible implementation and evaluate their practical performance through the Monte Carlo experiments. 

Assumption~\ref{ass:rate-nondeg} is a nondegeneracy condition for the scaled asymptotic risk problem. Under root-$T$ asymptotics, the unscaled denominator in the interior weight formula \eqref{eq:tildewstar} is typically of order $T^{-1}$; hence the relevant object is the scaled denominator $A_{T,h}+D_{T,h}-2F_{T,h}$. Requiring this term to be bounded away from zero makes the limiting interior risk problem locally strictly convex and the weight map locally smooth. Together with the interiority condition $\widetilde w_h^\star\in[\underline w,1-\underline w]$, this ensures that the clipped oracle weight $w_h^\star$ is uniquely well defined and that clipping is asymptotically inactive, so $w_h^\star=\widetilde w_h^\star$ in the limit. The assumption excludes boundary cases, such as correctly specified designs in which the efficient VAR receives essentially all the weight, as well as cases where LP and VAR become asymptotically equivalent and $A_h+D_h-2F_h$ vanishes. These cases are not covered by the smooth first-order expansion used below, but the clipped finite-sample procedure remains operationally well defined. Algorithm~\ref{alg:SV-weight} includes a numerical safeguard for small estimated denominators. Moreover, when LP and VAR are nearly equivalent, the weight may be poorly identified, but the averaged estimator remains well behaved because any convex combination of two nearly coincident estimators delivers nearly the same value.

For the benchmark first-order theory, $\widehat w_h$ denotes the plug-in weight computed from the scaled variance-covariance components $(\widehat A_{T,h},\widehat D_{T,h},\widehat F_{T,h})$; the finite-sample implementation in Algorithm~\ref{alg:SV-weight} additionally includes bootstrap bias terms.

\begin{theorem}[Consistency and rate of AR-sieve-bootstrap weights]\label{thm:SV-weights-proof}
For any fixed \(h\), under Assumptions~\ref{ass:VARb3} and \ref{ass:rate-nondeg},
$\widehat w_h \xrightarrow{p} w_h^\star$.
If, in addition,
\begin{align}
\left\|
(\widehat A_{T,h},\widehat D_{T,h},\widehat F_{T,h})
-
(A_h,D_h,F_h)
\right\|
=
O_p(r_T)+O_p(B^{-1/2}),
\label{eq:scaled-risk-rate-thm}
\end{align}
for some deterministic sequence \(r_T\to0\), then
$|\widehat w_h-w_h^\star|=O_p(r_T)+O_p(B^{-1/2})$.
\end{theorem}

Let $\Omega$ be a $2\times 2$ covariance matrix.
Define
\begin{align}
\label{eq:V_h}
V_h(w,\Omega) =
 w^2\Omega_{11} + (1-w)^2\Omega_{22} + 2w(1-w)\Omega_{12}.
\end{align}

\begin{theorem}[Consistency and asymptotic normality of $\widehat\theta_h(\widehat w_h)$]\label{thm:theta-avg-asyn}
For any fixed $h$, under Assumptions~\ref{ass:VARb2}--\ref{ass:rate-nondeg},
$\widehat\theta_h(\widehat w_h)\xrightarrow{p}\theta_h$, and
\begin{align}
\sqrt{T}\big(\widehat\theta_h(\widehat w_h)-\theta_h\big)
\Rightarrow \mathcal N\!\Big(0,\, V_h(w_h^\star,\Omega_h^{(2)})\Big), \label{eq:asyn-normal-avg}
\end{align}
where $V_h(w_h^\star,\Omega_h^{(2)})$ is defined by equation (\ref{eq:V_h}) with $w=w_h^\star$ and $\Omega=\Omega_h^{(2)}$.
\end{theorem}

Theorem~\ref{thm:SV-weights-proof} shows that the AR-sieve-bootstrap plug-in weight 
$\widehat w_h$ is consistent for the limiting oracle weight $w_h^\star$ at each fixed horizon $h$, provided that the scaled bootstrap risk components consistently estimate their limiting counterparts. The scaling is important because, under the benchmark root-$T$ asymptotics, the unscaled variance and covariance terms entering the finite-sample MSE are of order $T^{-1}$. Thus, the relevant nondegeneracy condition is imposed on the scaled denominator $A_h+D_h-2F_h$, rather than on the raw denominator $a_{T,h}+d_{T,h}-2f_{T,h}$.
If a rate statement is desired, it depends on the convergence rate of the scaled risk-component estimator, denoted by $r_T$, together with the Monte Carlo simulation error $B^{-1/2}$. We leave $r_T$ as a high-level rate because its exact value depends on the AR-sieve approximation, the lag-order sequence, and the statistic being bootstrapped. Importantly, Theorem~\ref{thm:theta-avg-asyn} only requires $\widehat w_h \overset{p}{\longrightarrow} w_h^\star$, since the plug-in weight error is then first-order negligible.
Theorem~\ref{thm:theta-avg-asyn} shows that $\widehat\theta_h(\widehat w_h)$ is consistent for $\theta_h$ and asymptotically normal at the root-$T$ rate, with asymptotic variance $V_h(w_h^\star,\Omega_h^{(2)})$. Although Theorems~\ref{thm:SV-weights-proof}-\ref{thm:theta-avg-asyn} are derived for the benchmark setting in which both estimators are root-$T$ consistent and asymptotically centered, our simulations show that our estimator performs well also under misspecification.

\begin{theorem}[Consistency and rate of plug-in asymptotic variance]\label{thm:SV-Vhat-proof}

For any fixed $h$, let $\widehat V_h = V_h(\widehat w_h,\widehat\Omega_h^{(2)})$. Suppose that $\widehat\Omega_h^{(2)}\xrightarrow{p}\Omega_h^{(2)}$. Under Assumptions~\ref{ass:VARb2}--\ref{ass:rate-nondeg}, $\widehat V_h \xrightarrow{p} V_h(w_h^\star,\Omega_h^{(2)})$.
If, in addition, $|\widehat w_h-w_h^\star| = O_p(r_T)+O_p(B^{-1/2})$ and $\|\widehat\Omega_h^{(2)}-\Omega_h^{(2)}\| = O_p(s_T)+O_p(B^{-1/2})$ for some $s_T\to0$, then $|\widehat V_h-V_h(w_h^\star,\Omega_h^{(2)})|
=
O_p(r_T+s_T)+O_p(B^{-1/2})$.
\end{theorem}

Theorem~\ref{thm:SV-Vhat-proof} justifies plug-in inference based on
$\widehat V_h=V_h(\widehat w_h,\widehat\Omega_h^{(2)})$. In practice, one may estimate the $2\times2$ matrix
$\Omega_h^{(2)}$ by forming the stacked vector of estimating equations for
$(\widehat\theta_{\mathrm{LP},h},\widehat\theta_{\mathrm{VAR},h})$ and applying a standard HAC estimator (e.g., Newey--West) to the
resulting score/influence-function series. Equivalently, one can estimate $\Omega_h^{(2)}$ via the same AR-sieve
bootstrap used to construct $\widehat w_h$: compute the centered bootstrap draws $\widetilde Z_{T,h}^{*(b)}=Z_{T,h}^{*(b)}-\bar Z_{T,h}^\ast$ and set $\widehat\Omega_h^{(2)}=B^{-1}\sum_{b=1}^B \widetilde Z_{T,h}^{*(b)}\widetilde Z_{T,h}^{*(b)\prime}$.

As an alternative to plug-in Wald inference, one can use an AR-sieve bootstrap to obtain pointwise confidence intervals. Fit a stable reduced-form VAR($p$) sieve to $\{\boldsymbol{Y}_t\}$ and obtain residuals $\widehat v_t$. For each bootstrap draw, form wild innovations $v_t^{\ast(b)}=\eta_t^{(b)}\widehat v_t$, where $\eta_t^{(b)}$ are i.i.d.\ Rademacher multipliers taking values in $\{-1,+1\}$, and generate
\begin{equation}
\boldsymbol{Y}_t^{\ast(b)}
=
\sum_{j=1}^p \widehat A_j \boldsymbol{Y}_{t-j}^{\ast(b)}
+
v_t^{\ast(b)}
\end{equation}
recursively. When an external instrument is used, the same multiplier is applied to the instrument innovation so as to preserve the contemporaneous covariance between the instrument and the VAR residuals. Re-estimate LP and VAR on each pseudo-sample and recompute the averaged estimator to obtain $\widehat\theta_h^{\ast(b)}=\widehat\theta_h^{\ast(b)}(\widehat w_h^{\ast(b)})$. A centered symmetric bootstrap interval is $\big[\widehat\theta_h(\widehat w_h)-\delta_{\alpha,h},\ \widehat\theta_h(\widehat w_h)+\delta_{\alpha,h}\big]$, where $\delta_{\alpha,h}$ is the empirical $(1-\alpha)$ quantile of $\big|\widehat\theta_h^{\ast(b)}-\widehat\theta_h(\widehat w_h)\big|$ across $b=1,\ldots,B$.

\section{Monte Carlo Evidence}
\label{sec:MC}

To assess the finite-sample performance of our methods, we consider both a simple univariate design and a richer multivariate design. The univariate exercise transparently illustrates the LP--VAR bias--variance trade-off and evaluates how well the feasible procedures approximate the oracle weights. The multivariate exercise then assesses the same methods in a more realistic macroeconomic environment.

The simulations cover three environments. First, we consider correctly specified benchmark designs in which the finite-order VAR approximation is exact, so both LP and VAR consistently estimate the same impulse responses asymptotically. Second, we consider local-misspecification designs in which the DGP is local to a finite-order VAR. In this setting, LP remains first-order robust, while the VAR estimator may exhibit local asymptotic bias (see, e.g., \citealp{MontielOleaEtAl2026,NemtyrevBoldea2026}).\footnote{This is the environment most closely aligned with TLP. The comparison should therefore be viewed as complementary: TLP is tailored to local misspecification, while our estimator-averaging approach targets finite-sample IRF risk more broadly.} Third, we consider fixed-misspecification designs in which the finite-order VAR approximation remains misspecified as the sample size increases. In these cases, the VAR bias persists asymptotically, while finite-lag LP may also be affected by lag truncation. The latter two environments therefore fall outside the benchmark asymptotic framework developed in Section~\ref{sec:asy_theory}, which focuses on the correctly specified case where both LP and VAR consistently estimate the same impulse responses. The purpose of these additional simulations is to examine whether finite-sample MSE-based estimator averaging continues to adapt even in environments not covered by the formal asymptotic theory.\footnote{In both the univariate and multivariate exercises, we use fixed lag-selection rules. Hence, the estimators are not designed to be asymptotically equivalent as in \citet{PlagborgMollerWolf2021}.}

\subsection{A Univariate ARMA Design}

We begin with a simple univariate DGP. This design allows us to (i) verify that the infeasible oracle estimator averaging behaves as predicted by the LP--VAR bias--variance trade-off, (ii) evaluate how well the sieve-bootstrap implementation approximates the oracle weights in finite samples, and (iii) compare the risk properties of estimator averaging with those of $R^2$-based model averaging and TLP.

\subsubsection{Model and Estimators}
\label{subsec:MC_uni}
We consider the following univariate process, in the spirit of \citet{MontielOleaetal2025},
\begin{equation}
y_t
=
\rho y_{t-1}
+
\varepsilon_t
+
\alpha_T \varepsilon_{t-1},
\qquad
\varepsilon_t \overset{iid}{\sim}\mathcal N(0,1),
\label{eq:dgp}
\end{equation}
where $|\rho|<1$ and $\alpha_T$ controls the degree of departure from the correctly specified AR(1) benchmark. For a given value of $\alpha_T$, the true impulse response to a one-unit innovation is $\theta_0=1$, $\theta_1=\rho+\alpha_T$, and $\theta_h=\rho\theta_{h-1}$ for $h\ge2$, or equivalently $\theta_h=\rho^h+\alpha_T\rho^{h-1}$ for $h\ge1$. We evaluate horizons $h\in\{1,\ldots,H_{\max}\}$ and set $H_{\max}=10$.

We consider two types of exercises. In the fixed-$T$ horizon-profile analysis, we set $T=240$, an empirically relevant sample size used in \citet{MontielOleaetal2025}, and consider the correctly specified AR(1) case, $\alpha_T=0$, together with fixed ARMA(1,1) misspecification, $\alpha_T=\alpha$, where $\alpha\in\{0.5,0.9\}$. Since $T$ is fixed in this exercise, there is no separate local-to-AR case. In the finite-sample convergence analysis, where $T$ varies over $\{200,400,800\}$, we additionally consider a local-to-AR design. Specifically, we set $\alpha_T=\alpha\sqrt{200}\,T^{-1/2}$, with $\alpha\in\{0.5,0.9\}$. This is a local-to-zero sequence for the MA coefficient, with local constant $\alpha\sqrt{200}$. The normalization by $200$ makes $\alpha_T=\alpha$ at the baseline sample size $T=200$, so the local and fixed-misspecification designs coincide at $T=200$ but differ as $T$ increases.\footnote{Since the MA component shrinks at rate $T^{-1/2}$, the induced bias is of order $T^{-1/2}$ and its squared contribution is of order $T^{-1}$, the same order as the sampling variance. This makes the design useful for studying the local bias--variance trade-off emphasized in the TLP literature.}

For each exercise, we run 1,000 Monte Carlo replications. Each replication uses a burn-in period of 200 observations before retaining a sample of length $T$. At each horizon, we report results for six estimators: (1) the local projection estimator (LP), which is the OLS coefficient on $y_t$ in the regression of $y_{t+h}$ on $y_t$ and $y_{t-1}$, as in \citet{MontielOleaetal2025}; (2) the VAR estimator, implemented as an AR(1); (3) the infeasible oracle averaged estimator based on MSE minimization; (4) the sieve-bootstrap plug-in estimator-averaging procedure; (5) the TLP shrinkage estimator; and (6) the $R^2$-based model-averaging estimator. To approximate the oracle weight and implement the feasible averaging procedures, we use 500 bootstrap draws.

\subsubsection{Horizon-Specific Risk and Weights at Fixed Sample Size}

We first fix $T=240$ and study horizons $h=1,\ldots,H_{\max}$. This allows us to trace how the LP--VAR bias--variance trade-off evolves across horizons and to assess whether the oracle and feasible averaging rules can capture these complex dynamics.

\begin{figure}[!htb]
  \centering
  \begin{subfigure}[t]{0.4\textwidth}
    \centering
   \includegraphics[width=\textwidth]{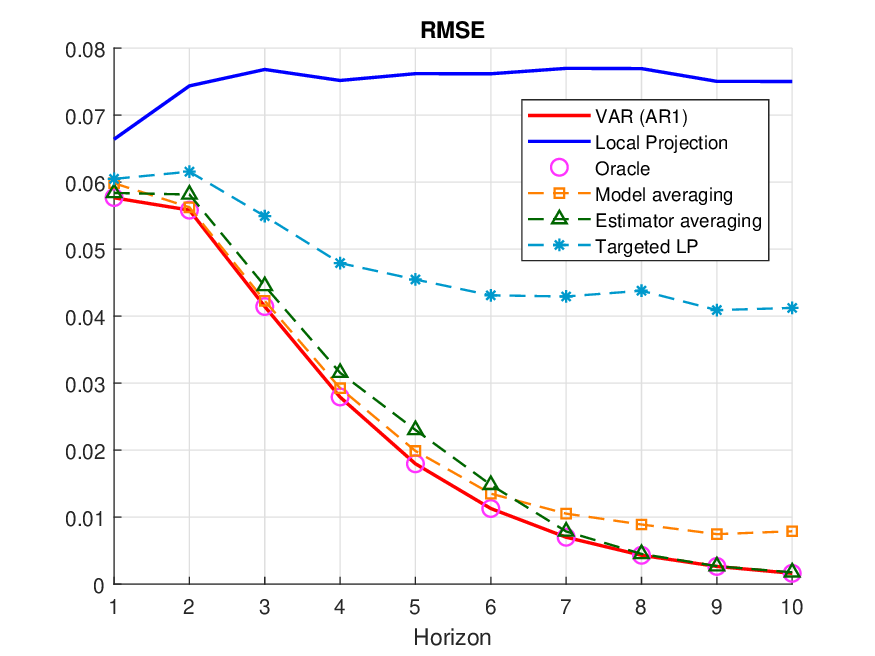}
  \end{subfigure}
  \hspace{0.04\textwidth}
  \begin{subfigure}[t]{0.4\textwidth}
    \centering
   \includegraphics[width=\textwidth]{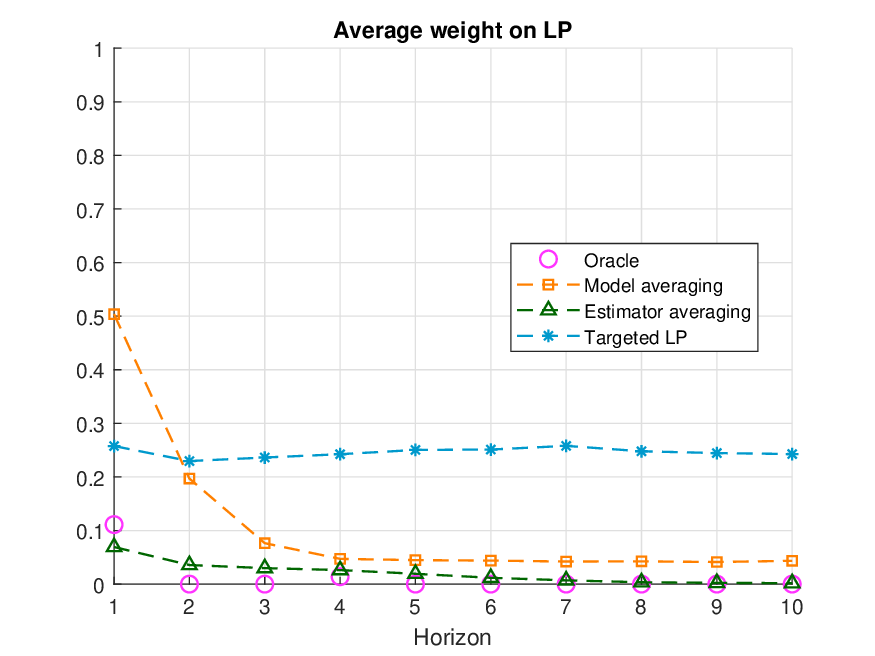}
  \end{subfigure}
  \caption{RMSE of IRF estimators and average weights on LP. AR(1) process with $\rho=0.5$, $T=240$, 1,000 replications.}
  \label{fig:arma50}
\end{figure}

Under the AR(1) DGP ($\rho=0.5, \alpha=0$), the VAR estimator is correctly specified and efficient, effectively coinciding with the oracle RMSE; see Figure~\ref{fig:arma50}. LP yields much higher RMSEs across all horizons, meaning the optimal oracle weight on LP is approximately zero. The feasible plug-in estimator averaging successfully captures this, correctly assigning near-zero weight to the LP and maintaining RMSEs very close to the oracle. Model averaging also performs well here, correctly favoring the VAR (particularly for $h>2$) and achieving similar risk reductions. Because TLP is designed as a shrinkage estimator under local misspecification, it does not fully collapse to the efficient VAR benchmark in this correctly specified AR(1) design. It therefore assigns positive weight to LP even though VAR is the lower-risk estimator here, leading to higher RMSEs than the plug-in and model-averaging approaches.

\begin{figure}[!htb]
  \centering
  \begin{subfigure}[t]{0.4\textwidth}
    \centering
   \includegraphics[width=\textwidth]{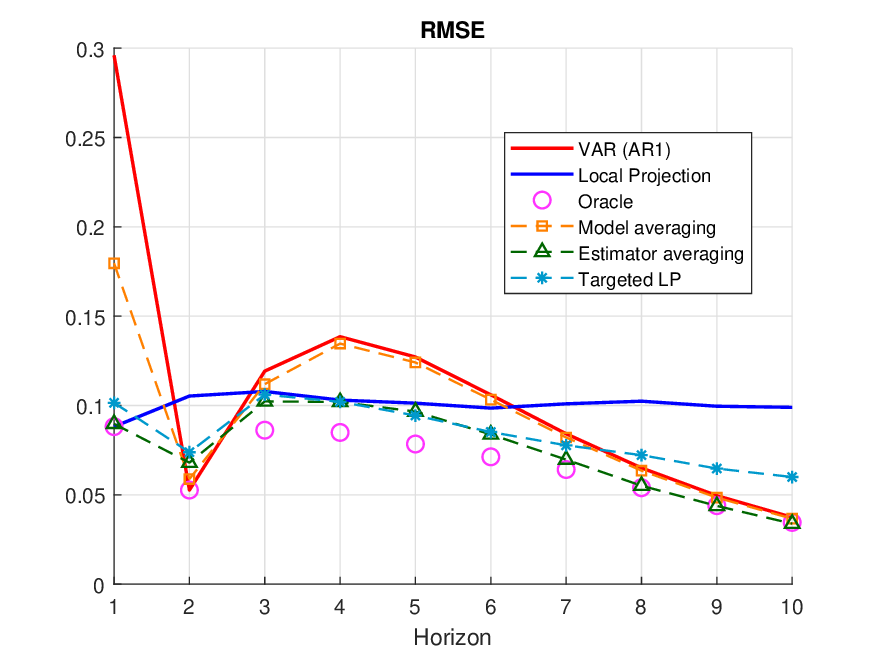}
  \end{subfigure}
  \hspace{0.04\textwidth}
  \begin{subfigure}[t]{0.4\textwidth}
    \centering
   \includegraphics[width=\textwidth]{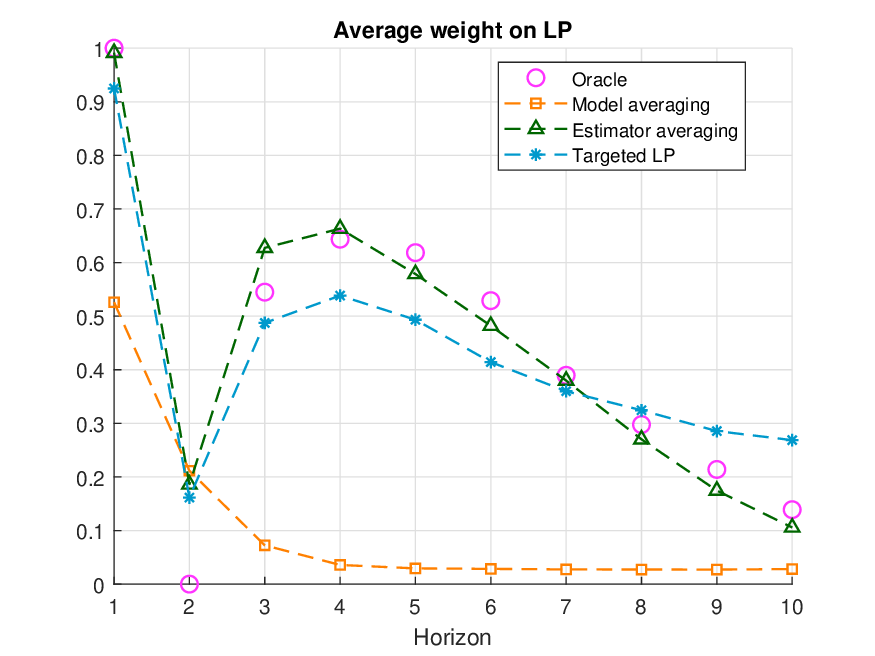}
  \end{subfigure}
  \caption{RMSE of IRF estimators and average weights on LP. ARMA(1,1) process with $\rho=0.5$ and $\alpha=0.5$, $T=240$, 1,000 replications.}
  \label{fig:arma55}
\end{figure}

Figures~\ref{fig:arma55} and \ref{fig:arma59} report the RMSEs and weights for the ARMA(1,1) DGPs, under which the finite-order baseline estimators are misspecified. In the case with moderate moving-average dynamics ($\alpha=0.5$), the optimal LP weight exhibits a highly non-monotonic trajectory: it drops dramatically from $h=1$ to $h=2$, rises again, and only begins a steady decline after horizons 4 and 5. Because LP performs much better than VAR at these specific short horizons ($h=1$ and $h \in [3,6]$), it commands a higher optimal weight. Here, both the plug-in estimator averaging and the TLP track this dynamic much better than model averaging, resulting in significantly lower RMSEs. At longer horizons, however, the TLP weights move less closely with the oracle trajectory, while the plug-in estimator averaging remains closer to the risk-minimizing benchmark.

\begin{figure}[!htb]
\centering
\begin{subfigure}[t]{0.4\textwidth}
\centering
\includegraphics[width=\textwidth]{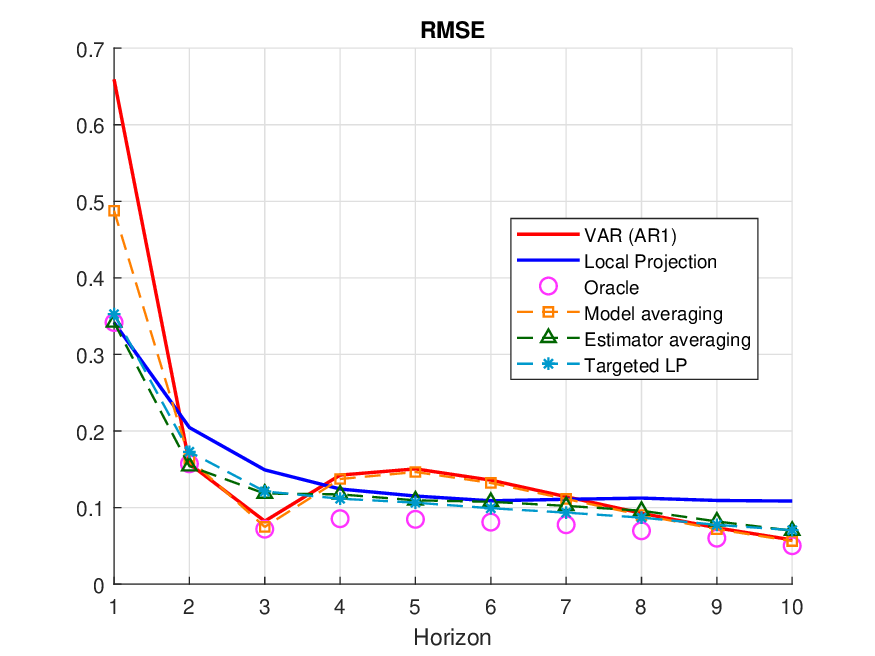}
\end{subfigure}
\hspace{0.04\textwidth}
\begin{subfigure}[t]{0.4\textwidth}
\centering
\includegraphics[width=\textwidth]{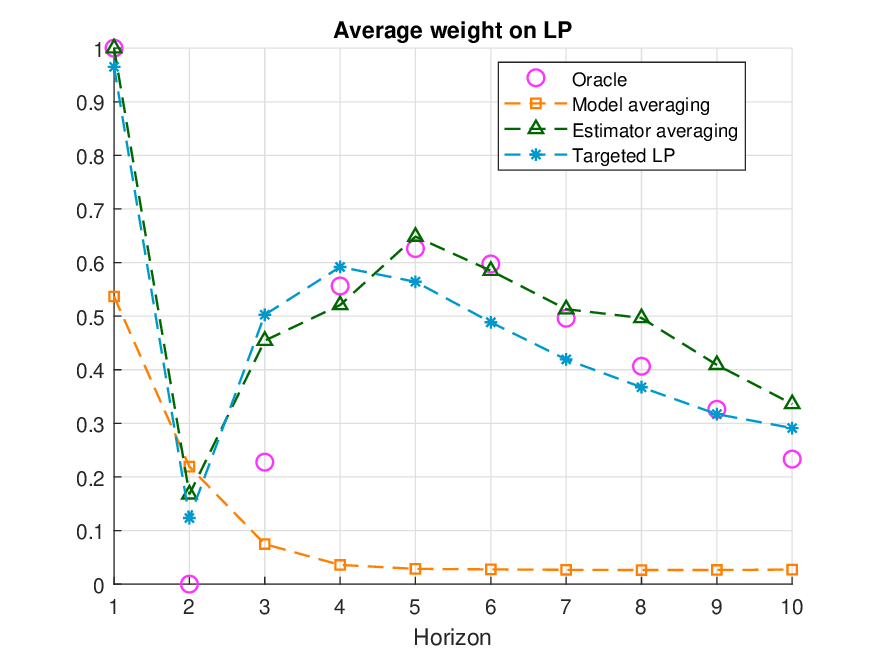}
\end{subfigure}
\caption{RMSE of IRF estimators and average weights on LP. ARMA(1,1) process with $\rho=0.5$ and $\alpha=0.9$, $T=240$, 1,000 replications.}
\label{fig:arma59}
\end{figure}

Under the ARMA(1,1) DGP with $\alpha=0.9$, the relative RMSEs of the VAR and LP estimators, and consequently the shape of the optimal weights across horizons, are remarkably similar to the $\alpha=0.5$ case. While the qualitative dynamics remain unchanged, the performance gap between the plug-in estimator averaging and the TLP largely disappears, and both approaches generally outperform model averaging across most horizons. 

The persistent underperformance of model averaging under these DGPs stems from the rigidity of its weighting scheme. Relying mechanically on in-sample fit ($R^2$) rather than out-of-sample IRF estimation risk, it entirely fails to capture the strong horizon-dependence of the oracle weights. Instead of matching the complex drop-and-rise dynamics, it mechanically assigns a weight of almost 1 to the VAR estimator for all horizons $h \ge 3$. 

Overall, this fixed-sample exercise supports the finite-sample risk-based motivation of the plug-in estimator averaging approach. In the correctly specified AR(1) case, the method appropriately places little weight on LP and remains close to the efficient VAR benchmark. In the fixed ARMA(1,1) designs, where the bias--variance trade-off becomes more complex and horizon-specific, the plug-in estimator adjusts the LP weight in line with the oracle risk-minimizing benchmark. TLP is also competitive in these misspecified designs, although it is less well suited to the correctly specified AR(1) benchmark where VAR is already efficient. By contrast, the $R^2$-based model-averaging rule is less responsive to horizon-specific IRF risk, especially in the ARMA(1,1) designs.

\subsubsection{Finite-Sample Convergence}
\label{sec:More:MC:section}

To complement the horizon-profile evidence, we next examine finite-sample convergence by varying $T\in\{200,400,800\}$ and focusing on representative horizons $h\in\{1,3,6\}$. Table~\ref{tab:uni_rmse_0.5} reports the RMSEs of the IRF estimators for the baseline case with $\rho=0.5$. In addition to the correctly specified AR(1) design and the fixed ARMA(1,1) designs, this convergence exercise also includes local-to-AR designs, where the MA coefficient shrinks with $T$ as described in Section~\ref{subsec:MC_uni}. Both the local and fixed misspecification designs are evaluated at $\alpha\in\{0.5,0.9\}$.  Section~\ref{app:MCuni_tables} in the Online Appendix reports additional results: Table~\ref{tab:uni_rmse_0.9} provides the analogous IRF RMSEs for $\rho=0.9$, and Tables~\ref{tab:uni_weight_0.5}--\ref{tab:uni_weight_0.9} report the RMSEs of the estimated weights relative to the oracle.

\begin{sidewaystable}[!htbp]
\centering
\caption{RMSE of univariate IRF estimators with $\rho=0.5$}
\label{tab:uni_rmse_0.5}
\resizebox{\textwidth}{!}{%
\begin{tabular}{lcccccccccccccccccc}
\toprule
& \multicolumn{6}{c}{$h=1$} & \multicolumn{6}{c}{$h=3$} & \multicolumn{6}{c}{$h=6$} \\
\cmidrule(lr){2-7}\cmidrule(lr){8-13}\cmidrule(lr){14-19}
$T$ & $\widehat{\theta}_{\mathrm{VAR}}$ & $\widehat{\theta}_{\mathrm{LP}}$ & $\widehat{\theta}_O$ & $\widehat{\theta}_P$ & $\widehat{\theta}_{TLP}$ & $\widehat{\theta}_M$ & $\widehat{\theta}_{\mathrm{VAR}}$ & $\widehat{\theta}_{\mathrm{LP}}$ & $\widehat{\theta}_O$ & $\widehat{\theta}_P$ & $\widehat{\theta}_{TLP}$ & $\widehat{\theta}_M$ & $\widehat{\theta}_{\mathrm{VAR}}$ & $\widehat{\theta}_{\mathrm{LP}}$ & $\widehat{\theta}_O$ & $\widehat{\theta}_P$ &$\widehat{\theta}_{TLP}$ & $\widehat{\theta}_M$ \\
\midrule
\multicolumn{19}{c}{AR(1)}\\
\hline
200 & 0.0639 & 0.0723 & 0.0639 & 0.0653 & 0.0661 & 0.0657 & 0.0448 & 0.0834 & 0.0448 & 0.0485 & 0.0596 & 0.0456 & 0.0122 & 0.0843 & 0.0122 & 0.0124 & 0.0492 & 0.0163 \\
400 & 0.0435 & 0.0498 & 0.0435 & 0.0443 & 0.0454 & 0.0451 & 0.0322 & 0.0580 & 0.0323 & 0.0346 & 0.0427 & 0.0323 & 0.0086 & 0.0585 & 0.0086 & 0.0091 & 0.0334 & 0.0094 \\
800 & 0.0311 & 0.0350 & 0.0311 & 0.0314 & 0.0324 & 0.0319 & 0.0229 & 0.0426 & 0.0229 & 0.0248 & 0.0312 & 0.0232 & 0.0059 & 0.0426 & 0.0059 & 0.0072 & 0.0245 & 0.0062 \\
\addlinespace
\multicolumn{19}{c}{local-to-AR(1), $\alpha=0.5$}\\
\hline
200 & 0.2995 & 0.0962 & 0.0962 & 0.0988 & 0.1123 & 0.1847 & 0.1184 & 0.1154 & 0.0891 & 0.1103 & 0.1108 & 0.1109 & 0.1053 & 0.1126 & 0.0794 & 0.0889 & 0.0949 & 0.1027 \\
400 & 0.1814 & 0.0519 & 0.0519 & 0.0560 & 0.0607 & 0.1046 & 0.1047 & 0.0736 & 0.0648 & 0.0735 & 0.0778 & 0.0986 & 0.0755 & 0.0727 & 0.0528 & 0.0565 & 0.0624 & 0.0746 \\
800 & 0.1120 & 0.0352 & 0.0352 & 0.0402 & 0.0402 & 0.0646 & 0.0822 & 0.0497 & 0.0465 & 0.0508 & 0.0547 & 0.0777 & 0.0498 & 0.0495 & 0.0359 & 0.0374 & 0.0421 & 0.0495 \\
\addlinespace
\multicolumn{19}{c}{ARMA(1,1), $\alpha=0.5$}\\
\hline
200 & 0.2995 & 0.0962 & 0.0962 & 0.0988 & 0.1123 & 0.1847 & 0.1184 & 0.1154 & 0.0891 & 0.1103 & 0.1108 & 0.1109 & 0.1053 & 0.1126 & 0.0794 & 0.0889 & 0.0949 & 0.1027 \\
400 & 0.2914 & 0.0790 & 0.0790 & 0.0795 & 0.0871 & 0.1753 & 0.1174 & 0.0820 & 0.0685 & 0.0788 & 0.0847 & 0.1096 & 0.1044 & 0.0785 & 0.0626 & 0.0773 & 0.0743 & 0.1031 \\
800 & 0.2881 & 0.0725 & 0.0725 & 0.0725 & 0.0768 & 0.1729 & 0.1169 & 0.0598 & 0.0523 & 0.0570 & 0.0623 & 0.1086 & 0.1037 & 0.0567 & 0.0500 & 0.0634 & 0.0605 & 0.1029 \\
\addlinespace
\multicolumn{19}{c}{local-to-AR(1), $\alpha=0.9$}\\
\hline
200 & 0.6621 & 0.3476 & 0.3476 & 0.3476 & 0.3606 & 0.4917 & 0.0833 & 0.1585 & 0.0752 & 0.1261 & 0.1266 & 0.0761 & 0.1349 & 0.1245 & 0.0918 & 0.1206 & 0.1115 & 0.1313 \\
400 & 0.4077 & 0.1397 & 0.1397 & 0.1397 & 0.1467 & 0.2635 & 0.1140 & 0.0920 & 0.0678 & 0.0846 & 0.0888 & 0.1052 & 0.1229 & 0.0827 & 0.0677 & 0.0866 & 0.0814 & 0.1212 \\
800 & 0.2484 & 0.0574 & 0.0574 & 0.0576 & 0.0620 & 0.1456 & 0.1143 & 0.0574 & 0.0518 & 0.0559 & 0.0612 & 0.1066 & 0.0945 & 0.0554 & 0.0482 & 0.0597 & 0.0580 & 0.0938 \\
\addlinespace
\multicolumn{19}{c}{ARMA(1,1), $\alpha=0.9$}\\
\hline
200 & 0.6621 & 0.3476 & 0.3476 & 0.3476 & 0.3606 & 0.4917 & 0.0833 & 0.1585 & 0.0752 & 0.1261 & 0.1266 & 0.0761 & 0.1349 & 0.1245 & 0.0918 & 0.1206 & 0.1115 & 0.1313 \\
400 & 0.6555 & 0.3383 & 0.3383 & 0.3383 & 0.3448 & 0.4851 & 0.0774 & 0.1250 & 0.0568 & 0.0918 & 0.1032 & 0.0687 & 0.1348 & 0.0872 & 0.0700 & 0.0910 & 0.0853 & 0.1329 \\
800 & 0.6526 & 0.3362 & 0.3362 & 0.3362 & 0.3394 & 0.4834 & 0.0753 & 0.1048 & 0.0433 & 0.0685 & 0.0898 & 0.0656 & 0.1350 & 0.0630 & 0.0544 & 0.0718 & 0.0668 & 0.1339 \\
\bottomrule
\end{tabular}
}
\begin{minipage}{\textwidth}\footnotesize
Note: $\widehat{\theta}_{\mathrm{VAR}}$ and $\widehat{\theta}_{\mathrm{LP}}$ are the VAR and LP IRF estimators.
$\widehat{\theta}_O$ uses the infeasible oracle min-MSE weight $w_h^\star$ in \eqref{eq:wstar}.
$\widehat{\theta}_P$ uses the sieve-bootstrap plug-in weight $\widehat w_P$ (Algorithm~\ref{alg:SV-weight}).
$\widehat{\theta}_{TLP}$ is the TLP estimator. $\widehat{\theta}_M$ uses the $R^2$-based model-averaging weight $\widehat w_{M}$ in \eqref{eq:wr2}. The autoregressive coefficient ($\rho$) is 0.5 in each case. Results are based on 1,000 Monte Carlo replications and 500 bootstrap iterations.
\end{minipage}
\end{sidewaystable}

Two main messages emerge from the IRF estimation results. First, the RMSEs of the feasible plug-in estimator averaging generally decline as $T$ increases and remain close to the oracle benchmark. Second, the plug-in approach is consistently competitive and often delivers the lowest RMSE among the feasible methods. TLP becomes more competitive when misspecification is more relevant, consistent with its shrinkage motivation under local misspecification, while the $R^2$-based model-averaging rule is less responsive to horizon-specific IRF risk.

The relative performance of the alternative feasible methods depends heavily on the extent to which the finite-order baseline estimators are misspecified. When the true DGP features a small or zero moving-average component (e.g., the AR(1) case), model averaging performs adequately, whereas the TLP yields the highest RMSEs. Conversely, under DGPs inducing severe misspecification ($\alpha=0.9$), the performance of the TLP improves relative to the alternatives, while model averaging performs less well in these designs. The plug-in estimator averaging is the most consistently adaptive feasible method across these shifting environments. 

Finally, the weight RMSE results reported in Section~\ref{app:MCuni_tables} of the Online Appendix are broadly consistent with our theoretical framework. Theorem~\ref{thm:SV-weights-proof} establishes weight consistency strictly for environments where the underlying finite-order estimators are correctly specified. In line with this, the estimated weights reliably converge toward the oracle benchmark under the AR(1) DGP, with only one exception ($h=6$ when $\rho=0.5$). In designs where the baseline estimators are formally misspecified (the local-to-AR and fixed ARMA cases), we lack analogous theoretical guarantees, and strict weight convergence is less evident in the simulations. However, the IRF RMSEs confirm that the weaker evidence of weight convergence---which likely reflects weak identification of the optimal weight in finite samples---has negligible practical consequences. The plug-in approach remains highly resilient, successfully exploiting the finite-sample bias--variance trade-off to minimize IRF estimation risk regardless of whether the estimated weights themselves achieve asymptotic consistency.

\subsection{A Multivariate SVARMA Design}
\label{subsec:MC_multi}

We next extend the analysis to a multivariate setting that more closely resembles a macroeconomic impulse-response application. We consider a three-variable system that nests a correctly specified SVAR(4) benchmark and an SVARMA(4,1) extension. The latter is used to study both local and fixed misspecification: in the local case, the moving-average component shrinks with the sample size, while in the fixed case it remains constant. In the SVARMA designs, the finite-order VAR omits the MA component, while finite-lag LP may also be affected by lag truncation. Throughout, we assume that the structural shocks are observed, allowing us to abstract from identification issues and focus directly on impulse-response estimation.

\subsubsection{Model and Estimators}

We consider two multivariate DGP families. The first is a correctly specified SVAR(4) benchmark with coefficient matrices $(A_j^s,M_0^s)$. The second is an SVARMA(4,1) family with coefficient matrices $(A_j^m,M_0^m,M_1^m)$:
\[
Y_t=\sum_{j=1}^4 A_j^mY_{t-j}+M_0^m\varepsilon_t+c_T M_1^m\varepsilon_{t-1},
\qquad
\varepsilon_t \sim \mathcal{N}(0,I_3).
\]

Setting $c_T=1$ gives the fixed-misspecification SVARMA(4,1) design, while setting $c_T=c_0T^{-1/2}$ gives the local-to-SVAR design. The coefficient matrices are reported in Section~\ref{app:multi_details} of the Online Appendix.

We conduct two exercises. For the fixed-$T$ horizon-profile analysis, we set $T=200$ and compare the correctly specified SVAR(4) design with the fixed-misspecification SVARMA(4,1) design. For the finite-sample convergence analysis, where $T\in\{200,800,2000\}$, we additionally include the local-to-SVAR design. In that case, we set $c_0=\sqrt{200}$, so the local and fixed SVARMA designs coincide at the baseline sample size $T=200$, while the MA component shrinks at rate $T^{-1/2}$. 

In parallel with the univariate case, we evaluate six estimators: (1) LP, with the lag length set equal to the lag order used for the VAR; (2) VAR, with the lag order selected by AIC up to a maximum of 8 lags; (3) the infeasible oracle averaged estimator based on MSE minimization; (4) the feasible plug-in estimator-averaging procedure, with weights estimated by the AR-sieve bootstrap; (5) the TLP shrinkage estimator of \citet{NemtyrevBoldea2026}; and (6) the $R^2$-based model-averaging estimator. For each exercise, we run 1,000 Monte Carlo replications. We use 500 bootstrap replications to approximate the oracle weight and implement the feasible averaging procedures.

\subsubsection{Horizon-Specific Risk and Weights at Fixed Sample Size}

Figure~\ref{fig:multi_results} reports the RMSEs and estimated weights for the SVAR(4) and SVARMA(4,1) DGPs when $T=200$. In the correctly specified SVAR(4) case, the bias--variance trade-off is clearly visible; see Panels (a) and (c). LP delivers lower RMSEs for horizons $h=2$ to $h=5$, after which VAR becomes superior because its structured dynamics produce greater efficiency. The oracle estimator closely tracks the lower envelope of the two individual estimators.

\begin{figure}[H]
\centering
\begin{subfigure}[t]{0.4\textwidth}
\centering
\includegraphics[width=\textwidth]{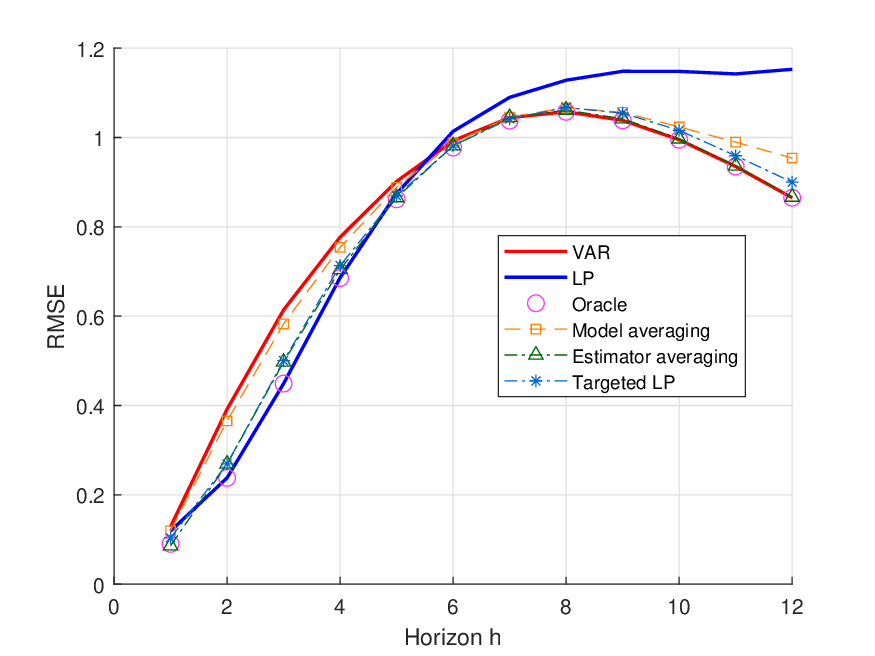}
\caption{RMSE: SVAR(4)}
\label{fig:rmseVAR4}
\end{subfigure}
\hspace{0.04\textwidth}
\begin{subfigure}[t]{0.4\textwidth}
\centering
\includegraphics[width=\textwidth]{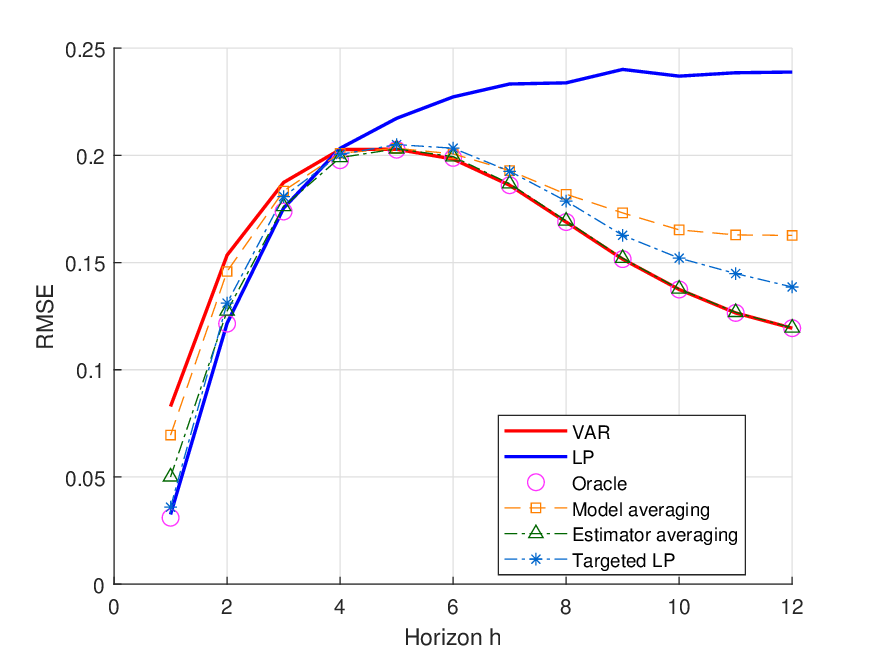}
\caption{RMSE: SVARMA(4,1)}
\label{fig:rmseVARMA41}
\end{subfigure}

\vspace{0.4em}

\begin{subfigure}[t]{0.4\textwidth}
\centering
\includegraphics[width=\textwidth]{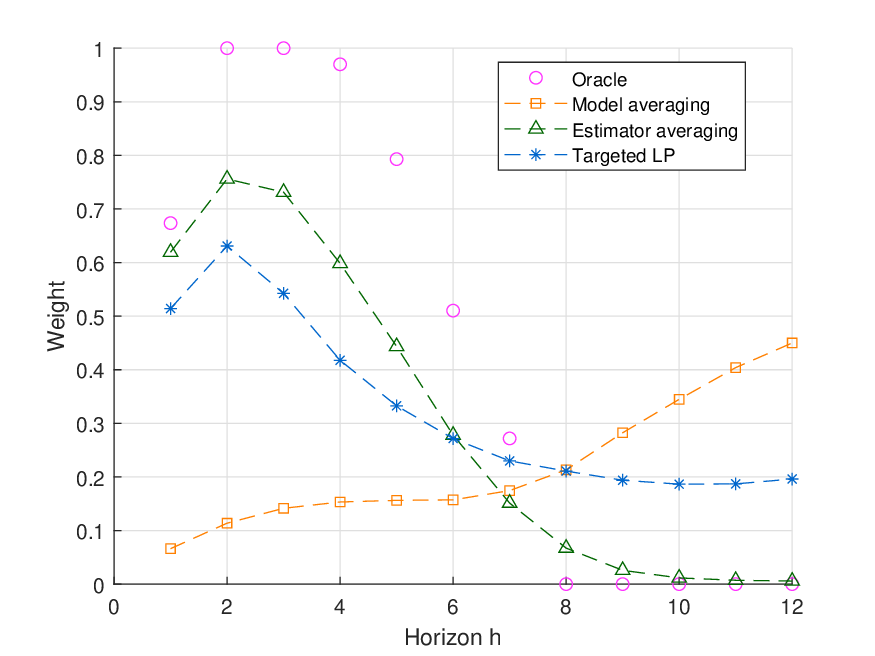}
\caption{Weights: SVAR(4)}
\label{fig:weightVAR4}
\end{subfigure}
\hspace{0.04\textwidth}
\begin{subfigure}[t]{0.4\textwidth}
\centering
\includegraphics[width=\textwidth]{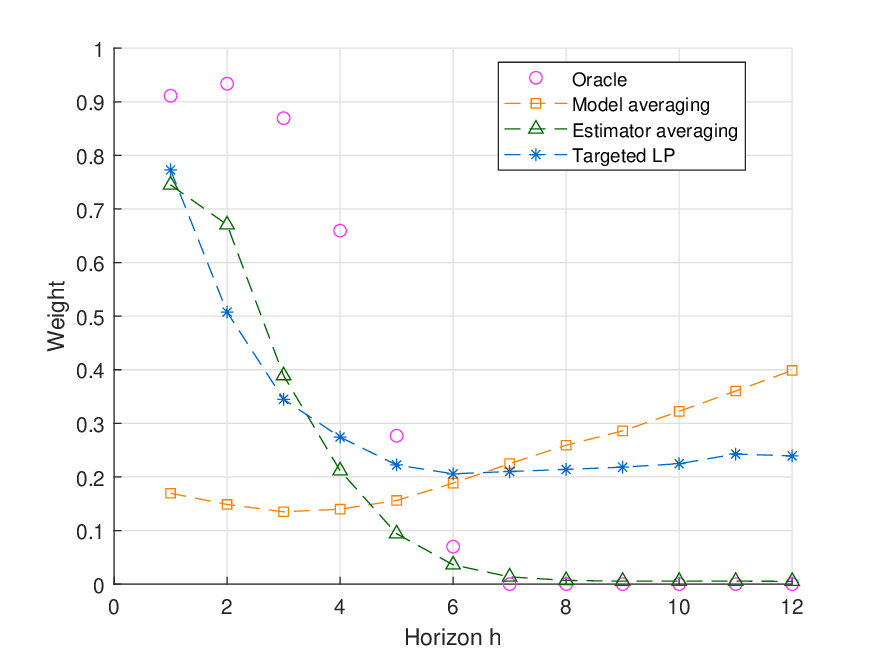}
\caption{Weights: SVARMA(4,1)}
\label{fig:weightVARMA41}
\end{subfigure}

\caption{Root mean squared errors (RMSE) and average weights on LP for two multivariate data-generating processes. Results are based on a sample size of $T=200$, using 1,000 Monte Carlo replications and 500 bootstrap iterations.}
\label{fig:multi_results}
\end{figure}

In terms of RMSE, the feasible plug-in estimator-averaging procedure is nearly indistinguishable from the oracle, indicating that the risk-based weighting rule is well estimated in finite samples. The TLP generally performs somewhat worse than the plug-in estimator but dominates model averaging. Specifically, the TLP achieves RMSEs that are almost identical to the plug-in approach at short and medium horizons; however, for horizons $h \ge 8$, its RMSE becomes larger. The weight plot in Panel (c) makes the reason clear: while the TLP tracks the dynamic shape of the oracle weights much better than model averaging, it does so less precisely than the plug-in estimator. By contrast, model averaging performs noticeably worse across the board, tending to underweight LP early and overweight it later because it is driven by in-sample fit rather than IRF estimation risk.

In the misspecified case (the SVARMA(4,1) DGP), despite both baseline estimators suffering from misspecification bias, LP retains its advantage at short horizons ($h<4$), while VAR performs better thereafter; see Panels (b) and (d). The oracle again tracks the minimum RMSE, and the plug-in estimator averaging remains remarkably close to that benchmark. The relative performance of the TLP is qualitatively similar to the correctly specified case, albeit with a slight short-run advantage: at impact ($h=1$) and intermediate horizons ($h=4$ and $5$), the TLP weights align closer to the oracle than the plug-in weights, yielding a better RMSE. At longer horizons, however, the advantage of the plug-in estimator becomes much more visible as it accurately captures the optimal weighting dynamics, which tilt entirely toward the VAR. Model averaging, meanwhile, continues to perform least effectively, failing to reproduce the rapid decline in the oracle LP weight and leading to substantially suboptimal averaging at longer horizons where the LP--VAR gap is large.

\subsubsection{Finite-sample convergence}
\label{sec:multi:asymp}

To summarize the multivariate convergence results, Table~\ref{tab:multi_irf_rmse} reports the RMSEs for $T\in\{200,800,2000\}$ at horizons $h\in\{1,6,12\}$.  Consistent with the univariate evidence, the RMSEs of estimator averaging decline as $T$ increases, remaining below those of both model averaging and the TLP across almost all specifications and sample sizes. Even at $T=2000$, where LP and VAR have both become quite accurate, estimator averaging typically yields a slight improvement by optimally combining their remaining finite-sample differences.

Table~\ref{tab:multi_weight_rmse} in the Online Appendix Section~\ref{app:DGP_matrices} reports the RMSE of the estimated weights relative to the oracle weights. Unlike the IRF estimates, the weights do not converge monotonically in all cases, particularly under SVARMA(4,1). This may be related to a flat or weakly identified weighting problem, which presumably occurs when the VAR and LP estimates converge toward each other. In such cases, estimator averaging can remain close to risk-optimal even when the estimated weights display discernible finite-sample variation.

\begin{sidewaystable}[!htbp]
\centering
\caption{RMSE of multivariate IRF estimators}
\label{tab:multi_irf_rmse}
\resizebox{\textwidth}{!}{%
\begin{tabular}{lcccccccccccccccccc}
\toprule
& \multicolumn{6}{c}{$h=1$} & \multicolumn{6}{c}{$h=6$} & \multicolumn{6}{c}{$h=12$} \\
\cmidrule(lr){2-7}\cmidrule(lr){8-13}\cmidrule(lr){14-19}
$T$ & $\widehat{\theta}_{\mathrm{VAR}}$ & $\widehat{\theta}_{\mathrm{LP}}$ & $\widehat{\theta}_O$ & $\widehat{\theta}_P$ & $\widehat{\theta}_{TLP}$ & $\widehat{\theta}_M$ & $\widehat{\theta}_{\mathrm{VAR}}$ & $\widehat{\theta}_{\mathrm{LP}}$ & $\widehat{\theta}_O$ & $\widehat{\theta}_P$ & $\widehat{\theta}_{TLP}$ & $\widehat{\theta}_M$ & $\widehat{\theta}_{\mathrm{VAR}}$ & $\widehat{\theta}_{\mathrm{LP}}$ & $\widehat{\theta}_O$ & $\widehat{\theta}_P$ &$\widehat{\theta}_{TLP}$ &  $\widehat{\theta}_M$ \\
\midrule
\multicolumn{19}{c}{SVAR(4)}\\
\hline
200  & 0.1260 & 0.1172 & 0.0869 & 0.0877 & 0.1030 & 0.1175 & 1.0213 & 1.0333 & 1.0005 & 1.0001 & 1.0129 & 1.0136 & 0.8737 & 1.1483 & 0.8737 & 0.8748 & 0.9107 & 0.9530 \\
800  & 0.0328 & 0.0531 & 0.0284 & 0.0271 & 0.0383 & 0.0309 & 0.4356 & 0.4476 & 0.4323 & 0.4318 & 0.4344 & 0.4345 & 0.4622 & 0.5520 & 0.4622 & 0.4629 & 0.4740 & 0.4878 \\
2000 & 0.0148 & 0.0330 & 0.0136 & 0.0130 & 0.0215 & 0.0140 & 0.2689 & 0.2839 & 0.2686 & 0.2677 & 0.2728 & 0.2688 & 0.3144 & 0.3475 & 0.3145 & 0.3141 & 0.3158 & 0.3181 \\
\addlinespace
\multicolumn{19}{c}{local-to-SVAR(4)}\\
\hline
200  & 0.0819 & 0.0320 & 0.0302 & 0.0442 & 0.0352 & 0.0685 & 0.1965 & 0.2245 & 0.1965 & 0.1974 & 0.2002 & 0.2000 & 0.1140 & 0.2435 & 0.1140 & 0.1142 & 0.1368 & 0.1638 \\
800  & 0.0214 & 0.0145 & 0.0122 & 0.0141 & 0.0143 & 0.0179 & 0.0983 & 0.1066 & 0.0983 & 0.0985 & 0.0992 & 0.0984 & 0.0629 & 0.1177 & 0.0629 & 0.0629 & 0.0711 & 0.0656 \\
2000 & 0.0090 & 0.0090 & 0.0064 & 0.0068 & 0.0077 & 0.0076 & 0.0634 & 0.0681 & 0.0635 & 0.0636 & 0.0639 & 0.0634 & 0.0420 & 0.0722 & 0.0420 & 0.0419 & 0.0457 & 0.0421 \\
\addlinespace
\multicolumn{19}{c}{SVARMA(4,1)}\\
\hline
200  & 0.0819 & 0.0320 & 0.0302 & 0.0442 & 0.0352 & 0.0685 & 0.1965 & 0.2245 & 0.1965 & 0.1974 & 0.2002 & 0.2000 & 0.1140 & 0.2435 & 0.1140 & 0.1142 & 0.1368 & 0.1638 \\
800  & 0.0216 & 0.0146 & 0.0122 & 0.0143 & 0.0143 & 0.0184 & 0.0949 & 0.1046 & 0.0949 & 0.0952 & 0.0965 & 0.0952 & 0.0590 & 0.1143 & 0.0590 & 0.0590 & 0.0678 & 0.0615 \\
2000 & 0.0093 & 0.0090 & 0.0065 & 0.0070 & 0.0078 & 0.0080 & 0.0612 & 0.0657 & 0.0616 & 0.0614 & 0.0619 & 0.0613 & 0.0367 & 0.0682 & 0.0367 & 0.0367 & 0.0413 & 0.0368 \\
\bottomrule
\end{tabular}
}
\begin{minipage}{\textwidth}\footnotesize
Note: $\widehat{\theta}_{\mathrm{VAR}}$ and $\widehat{\theta}_{\mathrm{LP}}$ are the VAR and LP IRF estimators.
$\widehat{\theta}_O$ uses the infeasible oracle min-MSE weight $w_h^\star$ in \eqref{eq:wstar}.
$\widehat{\theta}_P$ uses the sieve-bootstrap plug-in weight $\widehat w_P$ (Algorithm~\ref{alg:SV-weight}).
$\widehat{\theta}_{TLP}$ is the TLP estimator. $\widehat{\theta}_M$ uses the $R^2$-based model-averaging weight $\widehat w_{M}$ in \eqref{eq:wr2}. Results are based on 1,000 Monte Carlo replications and 500 bootstrap iterations.
\end{minipage}
\end{sidewaystable}

\section{Empirical Application}
\label{sec:application}
Section~V of \citet{BauerSwanson2023} re-assesses the dynamic effects of monetary policy by addressing two key challenges in high-frequency identification: instrument relevance and exogeneity. To improve relevance, they expand the standard set of monetary policy events beyond FOMC announcements to include press conferences, speeches, and testimony by the Federal Reserve Chair, substantially increasing the variation of the surprise series. To ensure exogeneity, they argue that conventional high-frequency surprises suffer from endogeneity because they are systematically correlated with publicly available macroeconomic and financial data predating the announcements---rather than being driven by central bank ``information effects.'' To address this, they orthogonalize the high-frequency surprises with respect to these pre-announcement variables and use the resulting residual as an external instrument for the monetary policy shock. They then estimate IV-LP and IV-SVAR impulse responses of yields, activity, prices, and financial conditions, finding that this correction produces stronger and more plausible macroeconomic estimates.

We replicate their baseline setup using the same monthly data set, orthogonalized high-frequency surprise, and IV specification. We study the responses of the two-year Treasury yield (GBY), industrial production (IP), consumer prices (CPI), and the excess bond premium (EBP) to a 25-basis-point contractionary monetary policy shock. For each variable, we compute IV-LP and IV-SVAR impulse responses based on their specifications, as well as model-averaging combinations of the two. Finally, we also compute the external-instrument version of the estimator-averaging combination described in Remark~\ref{remark:IV}. In the empirical application, we present only the plug-in version, as the flexible estimator averaging yields very similar estimates.

For inference, we employ a nested VAR-sieve wild bootstrap procedure. We first approximate the data-generating process by estimating a reduced-form VAR model, where the lag order is selected via the Bayesian information criterion. To handle potential heteroskedasticity and accommodate the missing observations in the instrument, we apply a wild bootstrap to the data. Specifically, to preserve the identifying contemporaneous correlation between the reduced-form VAR residuals and the external instrument, the same sequence of random wild multipliers---drawn from a Rademacher distribution---is applied simultaneously to both the VAR residuals and the instrument during resampling. Our procedure relies on a double bootstrap architecture in which both the outer and inner loops consist of 500 iterations. The inner bootstrap loop is utilized to approximate the finite-sample variances and biases of the individual IV-SVAR and IV-LP estimators, as well as their covariance, which provides the moments necessary to calculate the optimal weights for the combined estimator. Subsequently, the outer bootstrap loop generates the empirical distribution of all the estimators. The pointwise confidence intervals are computed by extracting the middle 68 percent of the bootstrap estimates centered around the point estimates from the original sample. These intervals should be interpreted as pointwise bootstrap intervals. A full theoretical analysis of the nested proxy-SVAR/IV-LP bootstrap with missing instrument observations is beyond the scope of the paper. Figure~\ref{fig:emp:IRFs} displays the resulting impulse responses. Additional results are reported in Online Appendix Section~\ref{app:IRFs}. Figure~\ref{fig:emp:weights} shows the corresponding weights on the IV-LP estimator, while Figures~\ref{fig:emp:GBY}--\ref{fig:emp:EBP} present the impulse responses with pointwise confidence intervals for all estimators.

Our baseline IV-VAR and IV-LP estimates successfully replicate the results documented in \citet{BauerSwanson2023} (see their Figures 3 and A2, respectively). The IV-VAR responses (red lines) are smooth and tightly shaped: GBY and EBP jump on impact and then gradually return toward zero, IP exhibits a modest hump-shaped decline, and CPI shows a prolonged disinflation. By contrast, the IV-LP responses (blue lines) are much more volatile. For GBY and EBP, LP IRFs oscillate and change sign several times; for IP the decline is steeper and very persistent; for CPI the response eventually turns positive and drifts upward, implying an implausible long-run rise in price level after a contractionary shock. This stark contrast between low-variance but potentially biased IV-VAR estimates and low-bias but noisy IV-LP estimates mirrors the bias–variance trade-off highlighted in our simulations and motivates the use of averaging estimators.

\begin{figure}[H]
  \centering
  \includegraphics[width=\textwidth]{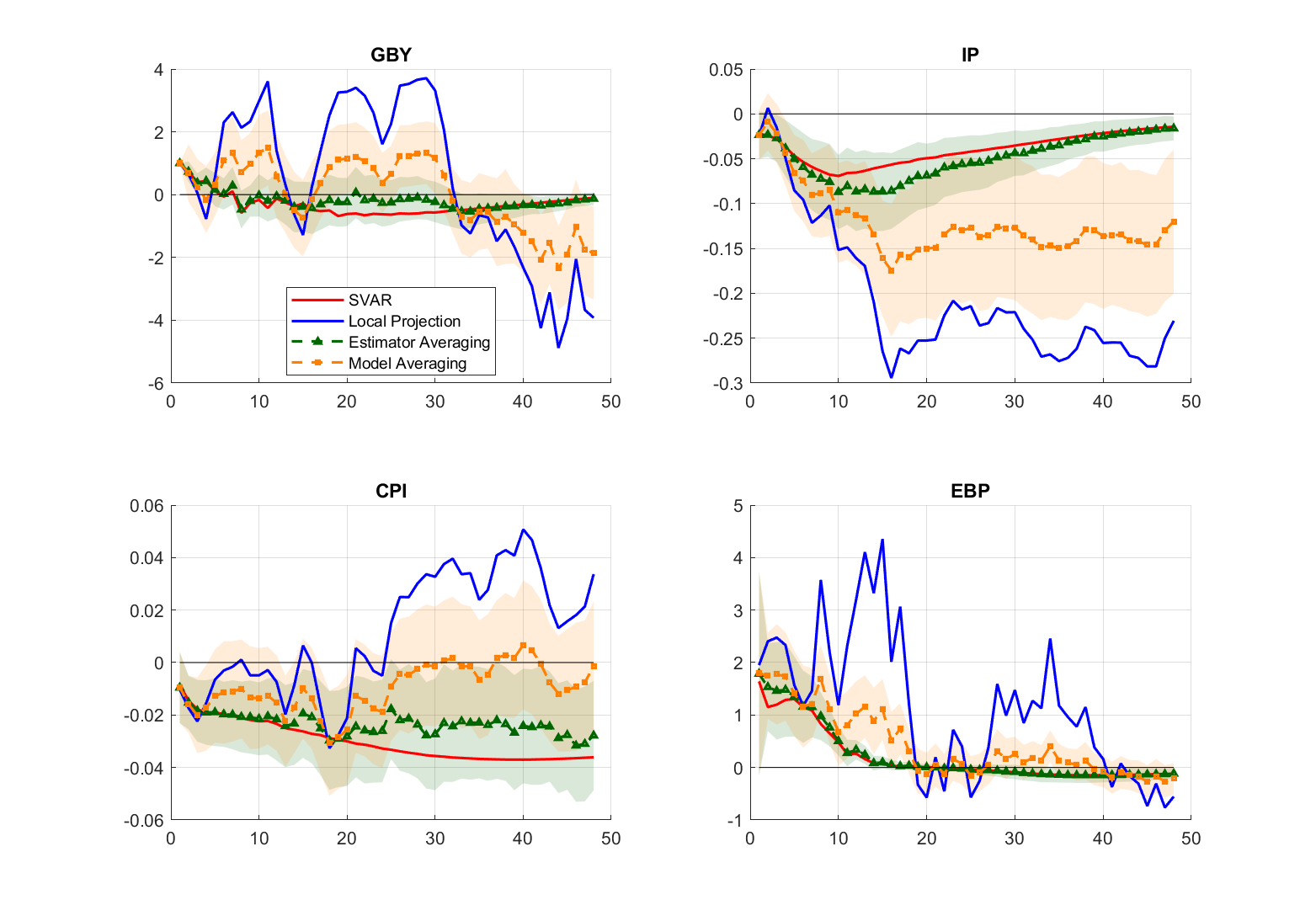}
  \caption{Point estimates of the impulse responses to a 25-basis-point contractionary monetary policy shock. Each panel compares the four estimation approaches for a given macroeconomic variable. The 68 percent pointwise confidence intervals for the averaging estimators are calculated using a nested VAR-sieve wild bootstrap procedure (500 inner and outer iterations) and are centered around the point estimates.}
  \label{fig:emp:IRFs}
\end{figure}

The averaging procedures stabilize the IRFs while preserving their main qualitative features. In many cases, they tilt toward the more economically plausible trajectory. Estimator averaging (green dashed lines with triangles) pulls the paths toward the smoother IV-VAR trajectories wherever the LP estimates are extremely noisy, but still allows for some deviations where the LP and VAR disagree. For GBY, the estimator-averaged IRF shows a sharp increase in the two-year yield that decays within a year, avoiding the large negative swings of the LP while not over-smoothing the near-term reaction. Furthermore, unlike the VAR response, it never sinks persistently into negative territory, a profile that is economically plausible following a tightening shock. For IP, estimator averaging yields a moderate, hump-shaped decline that lies between the highly persistent LP response and the more muted VAR response. For CPI, the combined estimator produces a smaller and less protracted disinflation than the VAR, while successfully avoiding the ``price puzzle'' anomaly exhibited by the LP. For EBP, estimator averaging delivers a sharp, short-run rise in credit spreads that quickly mean-reverts, eliminating the extreme volatility of the LP estimator while preserving the intuitive tightening in financial conditions after a monetary contraction. 

Model averaging based on $R^2$ (orange lines with squares) behaves differently. As Figure~\ref{fig:emp:weights} shows, the $R^2$-based weights on LP are relatively flat across horizons---around one-half for GBY, IP, and CPI, and somewhat lower for EBP. This reflects the fact that LP and VAR achieve similar in-sample fit even when the LP is estimated for longer horizons. Consequently, the model-averaged IRFs remain more heavily influenced by LP at longer horizons. For GBY and EBP, this yields more volatile responses than estimator averaging. The trajectory of the former drops substantially into negative territory after three years; for CPI, the effect disappears entirely after two years; and for IP, the decline is deeper and more persistent. Model averaging therefore improves on raw IV-LP by shrinking its most extreme movements, but it does not fully correct the long-horizon instability generated by the LP estimator and does not eliminate the puzzles.

Figure~\ref{fig:emp:weights} also makes clear why estimator averaging tends to deliver the most intuitive IRFs. The estimator-averaging weights on LP are high only on impact and in the very first few months, then quickly decay toward zero as the horizon increases, especially for IP and EBP. This pattern reflects the empirical fact that identification is strongest and LP variance is smallest at very short horizons, whereas the VAR's parametric structure provides more reliable long-run dynamics. The resulting estimator-averaged IRFs are therefore economically appealing: a front-loaded, temporary rise in GBY; a transitory fall in IP; a small and delayed disinflation in CPI; and a pronounced but short-lived increase in EBP. These responses sit between IV-LP and IV-VAR where the two disagree, dampen LP's long-horizon noise, and avoid relying mechanically on the smooth VAR trajectory, illustrating the practical usefulness of estimator averaging in applied monetary policy analysis.

\section{Conclusion}
\label{sec:conclusion}
LP and VAR are the two workhorse methods for impulse response analysis, and their relative appeal is fundamentally a finite-sample question. LP is often attractive because of its robustness and comparatively low bias, whereas VAR is often attractive because of its greater precision. This paper studies how to combine these two estimators through horizon-specific estimator averaging, with weights chosen to minimize the mean squared error of the structural impulse response itself rather than the in-sample fit of the underlying regression.

We derive closed-form oracle weights that make transparent how the optimal LP share depends on the relative bias, variance, and covariance of LP and VAR, and we develop feasible AR-sieve-bootstrap implementations to estimate these weights in practice. The Monte Carlo results show that estimator averaging can deliver meaningful MSE gains relative to LP and VAR alone, especially because its inherent flexibility allows it to precisely track the horizon-specific dynamics of the bias--variance trade-off. In contrast, the fit-based model-averaging approach is less flexible and performs worse in our design.

In an empirical application revisiting the high-frequency IV monetary policy shocks of \citet{BauerSwanson2023}, estimator averaging yields IRFs for yields, activity, prices, and credit spreads that are stable, economically intuitive, and lie between the often volatile IV-LP estimates and the very smooth IV-VAR estimates. Overall, our results suggest that estimator averaging provides a practical and easy-to-implement complement to existing LP and VAR practice, especially for empirical researchers who want to discipline the finite-sample bias--variance trade-off directly at the level of the impulse response of interest.

\singlespacing
\bibliographystyle{apalike}
\bibliography{references} 

\doublespacing
\newpage
\begin{appendices}
\section{Lemmas}
\label{app:Lemma}

\begin{lemma}[Consistency and rate of scaled bootstrap risk components]
\label{lem:SV-risk-proof}
For any fixed $h$, define the scaled bootstrap variance-covariance components by 
$\widehat A_{T,h}=T\widehat V_{\mathrm{LP},h}^{SV}$, 
$\widehat D_{T,h}=T\widehat V_{\mathrm{VAR},h}^{SV}$, and 
$\widehat F_{T,h}=T\widehat C_h^{SV}$. Under Assumption~\ref{ass:VARb3},
\begin{align}
(\widehat A_{T,h},\widehat D_{T,h},\widehat F_{T,h})
\xrightarrow{p}
(A_h,D_h,F_h). 
\label{eq:scaled-bootstrap-risk-consistency}
\end{align}
If, in addition, the scaled bootstrap risk components satisfy
\begin{align}
\left\|
(\widehat A_{T,h},\widehat D_{T,h},\widehat F_{T,h})
-
(A_h,D_h,F_h)
\right\|
=
O_p(r_T)+O_p(B^{-1/2}),
\label{eq:scaled-bootstrap-risk-rate}
\end{align}
for some deterministic sequence $r_T\to0$, then the scaled risk components converge at that rate.
\end{lemma}

\begin{proof}[Proof of Lemma~\ref{lem:SV-risk-proof}]
Before proving the lemma, we clarify the role of the AR-sieve pseudo-truth. 
Under the short-memory Wold assumptions and the AR-sieve lag-order condition, 
with $p_T\to\infty$ sufficiently slowly relative to $T$, the AR-sieve approximation 
targets the same population impulse response in the benchmark setting where LP and 
VAR are both consistent. This motivates the finite-sample bootstrap bias correction 
used in Algorithm~\ref{alg:SV-weight}. However, first-order bootstrap validity does 
not by itself imply consistency of the bootstrap bias estimator, which is a higher-order 
object. Therefore, Lemma~\ref{lem:SV-risk-proof} only uses AR-sieve validity to establish 
consistency of the scaled bootstrap variance-covariance components. A formal consistency 
result for the bootstrap bias terms would require additional higher-order conditions, which 
we do not impose.

Let $Z_{T,h}^\ast=\sqrt{T}\big((\widehat\theta_{\mathrm{LP},h}^\ast,\widehat\theta_{\mathrm{VAR},h}^\ast)'-(\widehat\theta_{\mathrm{LP},h},\widehat\theta_{\mathrm{VAR},h})'\big)$, and write $Z_{T,h}^\ast=(Z_{L,T,h}^\ast,Z_{V,T,h}^\ast)'$. The scaled bootstrap variance and covariance estimators can be written as sample second moments of the bootstrap draws:
\begin{align}
\widehat A_{T,h}
&=
\frac{1}{B}\sum_{b=1}^B
\left(
Z_{L,T,h}^{\ast(b)}
-
\overline Z_{L,T,h}^\ast
\right)^2, \label{eq:Ahat-scaled-bootstrap}\\
\widehat D_{T,h}
&=
\frac{1}{B}\sum_{b=1}^B
\left(
Z_{V,T,h}^{\ast(b)}
-
\overline Z_{V,T,h}^\ast
\right)^2, \label{eq:Dhat-scaled-bootstrap}\\
\widehat F_{T,h}
&=
\frac{1}{B}\sum_{b=1}^B
\left(
Z_{L,T,h}^{\ast(b)}
-
\overline Z_{L,T,h}^\ast
\right)
\left(
Z_{V,T,h}^{\ast(b)}
-
\overline Z_{V,T,h}^\ast
\right). \label{eq:Fhat-scaled-bootstrap}
\end{align}

Assumption~\ref{ass:VARb3}(i) implies that the conditional distribution of $Z_{T,h}^\ast$ consistently approximates the distribution of $Z_{T,h}$ in the bounded-Lipschitz metric. Together with the conditional moment bound in Assumption~\ref{ass:VARb3}(ii), this gives convergence of the relevant conditional quadratic moments. In particular,
\begin{align}
\mathbb E^\ast[(Z_{L,T,h}^\ast)^2] &\xrightarrow{p}\Omega_{11,h}, &
\mathbb E^\ast[(Z_{V,T,h}^\ast)^2] &\xrightarrow{p}\Omega_{22,h}, &
\mathbb E^\ast[Z_{L,T,h}^\ast Z_{V,T,h}^\ast] &\xrightarrow{p}\Omega_{12,h}. 
\label{eq:bootstrap-second-moments}
\end{align}
Equivalently, after recentering by the bootstrap mean, the conditional covariance matrix of $Z_{T,h}^\ast$ converges in probability to $\Omega_h^{(2)}$.

Conditional on the data, the bootstrap draws are i.i.d.\ across $b$. Hence, as $B\to\infty$, the conditional LLN gives convergence of the bootstrap sample second moments to their conditional expectations. Combining this with \eqref{eq:bootstrap-second-moments} yields
\begin{align}
(\widehat A_{T,h},\widehat D_{T,h},\widehat F_{T,h})
\xrightarrow{p}
(\Omega_{11,h},\Omega_{22,h},\Omega_{12,h})
=
(A_h,D_h,F_h).
\end{align}

The displayed rate is not implied by first-order bootstrap validity alone. If the high-level scaled-risk rate condition~\eqref{eq:scaled-bootstrap-risk-rate} holds, then the scaled bootstrap risk components converge at that rate. This rate is carried forward to the plug-in weight in Theorem~\ref{thm:SV-weights-proof}.
\end{proof}

\section{Proof of Theorems}
\label{app:proof_theorem}
\begin{proof}[\textbf{Proof of Theorem \ref{thm:SV-weights-proof}}]
Define $G(A,D,F)=(D-F)/(A+D-2F)$, and let $\bar G(x)=\min\{1,\max\{0,x\}\}$ denote clipping. Then $\widehat w_h=\bar G\!\left(G(\widehat A_{T,h},\widehat D_{T,h},\widehat F_{T,h})\right)$ and $w_h^\star=\bar G\!\left(G(A_h,D_h,F_h)\right)$.

By Lemma~\ref{lem:SV-risk-proof},
\begin{align}
(\widehat A_{T,h},\widehat D_{T,h},\widehat F_{T,h})
\xrightarrow{p}
(A_h,D_h,F_h). \label{eq:weight-proof-risk-components}
\end{align}
Assumption~\ref{ass:rate-nondeg} implies $A_h+D_h-2F_h\ge c>0$. Therefore, $G$ is continuous in a neighborhood of $(A_h,D_h,F_h)$. Since the limiting solution is interior, clipping is asymptotically inactive. The continuous mapping theorem gives $\widehat w_h\xrightarrow{p}w_h^\star$.

For the rate statement, suppose that \eqref{eq:scaled-risk-rate-thm} holds. Because $G$ is continuously differentiable in a neighborhood of $(A_h,D_h,F_h)$ and its denominator is bounded away from zero, a mean-value expansion yields
\begin{align}
|\widehat w_h-w_h^\star|
&\le
C
\left\|
(\widehat A_{T,h},\widehat D_{T,h},\widehat F_{T,h})
-
(A_h,D_h,F_h)
\right\| \nonumber\\
&=
O_p(r_T)+O_p(B^{-1/2}).
\end{align}
This proves the theorem.
\end{proof}

\begin{proof}[\textbf{Proof of Theorem \ref{thm:theta-avg-asyn}}]
Define $\widehat\theta_h(w)=w\widehat\theta_{\mathrm{LP},h}+(1-w)\widehat\theta_{\mathrm{VAR},h}$. Then
\begin{align}
\widehat\theta_h(\widehat w_h)-\theta_h
=
\big(\widehat\theta_h(w_h^\star)-\theta_h\big)
+
(\widehat w_h-w_h^\star)
(\widehat\theta_{\mathrm{LP},h}-\widehat\theta_{\mathrm{VAR},h}). \label{eq:theta-decomp-consistency}
\end{align}
For consistency, Assumption~\ref{ass:VARb2} gives $(\widehat\theta_{\mathrm{LP},h},\widehat\theta_{\mathrm{VAR},h})'\xrightarrow{p}(\theta_h,\theta_h)'$. Together with $\widehat w_h\xrightarrow{p}w_h^\star$ from Theorem~\ref{thm:SV-weights-proof}, this implies $\widehat\theta_h(\widehat w_h)\xrightarrow{p}\theta_h$.

For asymptotic normality, multiply \eqref{eq:theta-decomp-consistency} by $\sqrt T$:
\begin{align}
\sqrt T\big(\widehat\theta_h(\widehat w_h)-\theta_h\big)
&=
\sqrt T\big(\widehat\theta_h(w_h^\star)-\theta_h\big)
+
(\widehat w_h-w_h^\star)
\sqrt T(\widehat\theta_{\mathrm{LP},h}-\widehat\theta_{\mathrm{VAR},h}). \label{eq:theta-decomp-an}
\end{align}
By Assumption~\ref{ass:VARb2}, $\sqrt T(\widehat\theta_{\mathrm{LP},h}-\widehat\theta_{\mathrm{VAR},h})=O_p(1)$, and by Theorem~\ref{thm:SV-weights-proof}, $\widehat w_h-w_h^\star=o_p(1)$. Thus the second term in \eqref{eq:theta-decomp-an} is $o_p(1)$.

For the first term,
\begin{align}
\sqrt T\big(\widehat\theta_h(w_h^\star)-\theta_h\big)
=
e(w_h^\star)'
\sqrt T
\Big((\widehat\theta_{\mathrm{LP},h},\widehat\theta_{\mathrm{VAR},h})'-(\theta_h,\theta_h)'\Big),
\end{align}
where $e(w_h^\star)=(w_h^\star,1-w_h^\star)'$. By Assumption~\ref{ass:VARb2} and the continuous mapping theorem,
\begin{align}
\sqrt T\big(\widehat\theta_h(w_h^\star)-\theta_h\big)
\Rightarrow
\mathcal N\!\left(
0,\,
e(w_h^\star)'\Omega_h^{(2)}e(w_h^\star)
\right). \label{eq:fixed-weight-limit}
\end{align}
Since $e(w_h^\star)'\Omega_h^{(2)}e(w_h^\star)=V_h(w_h^\star,\Omega_h^{(2)})$, Slutsky's theorem gives the stated limit for $\widehat\theta_h(\widehat w_h)$.
\end{proof}

\begin{proof}[\textbf{Proof of Theorem \ref{thm:SV-Vhat-proof}}]
The map $(w,\Omega)\mapsto V_h(w,\Omega)$ is a polynomial in $w$ and the entries of $\Omega$, hence continuous. By Theorem~\ref{thm:SV-weights-proof}, $\widehat w_h\xrightarrow{p}w_h^\star$, and by assumption, $\widehat\Omega_h^{(2)}\xrightarrow{p}\Omega_h^{(2)}$. The continuous mapping theorem gives
\begin{align}
\widehat V_h \xrightarrow{p} V_h(w_h^\star,\Omega_h^{(2)}).
\end{align}

For the rate statement, suppose that $|\widehat w_h-w_h^\star|=O_p(r_T)+O_p(B^{-1/2})$ and $\|\widehat\Omega_h^{(2)}-\Omega_h^{(2)}\|=O_p(s_T)+O_p(B^{-1/2})$ for some $s_T\to0$. A mean-value expansion of $V_h(w,\Omega)$ gives
\begin{align}
|\widehat V_h- V_h(w_h^\star,\Omega_h^{(2)})|
&\le
C_1|\widehat w_h-w_h^\star|
+
C_2\|\widehat\Omega_h^{(2)}-\Omega_h^{(2)}\| \nonumber\\
&=
O_p(r_T+s_T)+O_p(B^{-1/2}).
\end{align}
\end{proof}

\section{Online Supplementary Appendix}

This appendix contains four sections. Section~\ref{app:MCuni_tables} reports additional univariate simulation results. Section~\ref{app:multi_details} documents the multivariate DGP. Section~\ref{app:DGP_matrices} provides additional multivariate simulation results. Section~\ref{app:IRFs} presents additional results for the empirical application, including the estimated weights on the IV-LP estimator and impulse response functions with pointwise confidence intervals.

\subsection{Additional Univariate Simulation Results}
\label{app:MCuni_tables}

\begin{table}[!htbp]
\centering
\caption{RMSE of univariate estimator averaging weights relative to oracle with $\rho=0.5$}
\label{tab:uni_weight_0.5}
\begin{tabular}{lccc}
\toprule
$T$ & $h=1$ & $h=3$ & $h=6$ \\
\midrule
\multicolumn{4}{c}{AR(1)}\\
\hline
200 & 0.2057 & 0.1399 & 0.0343 \\
400 & 0.1555 & 0.1230 & 0.0380 \\
800 & 0.1341 & 0.1122 & 0.0487 \\
\addlinespace
\multicolumn{4}{c}{local-to-AR(1) with $\alpha=0.5$}\\
\hline
200 & 0.0391 & 0.2427 & 0.2952 \\
400 & 0.1139 & 0.1570 & 0.2404 \\
800 & 0.1896 & 0.1713 & 0.2000 \\
\addlinespace
\multicolumn{4}{c}{ARMA(1,1) with $\alpha=0.5$}\\
\hline
200 & 0.0391 & 0.2427 & 0.2952 \\
400 & 0.0078 & 0.1657 & 0.3448 \\
800 & 0.0013 & 0.0913 & 0.3069 \\
\addlinespace
\multicolumn{4}{c}{local-to-AR(1) with $\alpha=0.9$}\\
\hline
200 & 0.0002 & 0.3826 & 0.3071 \\
400 & 0.0007 & 0.2146 & 0.3122 \\
800 & 0.0041 & 0.0815 & 0.3030 \\
\addlinespace
\multicolumn{4}{c}{ARMA(1,1) with $\alpha=0.9$}\\
\hline
200 & 0.0002 & 0.3826 & 0.3071 \\
400 & 0.0000 & 0.2818 & 0.2983 \\
800 & 0.0000 & 0.2079 & 0.2825 \\
\bottomrule
\end{tabular}
\end{table}

\begin{table}[!htbp]
\centering
\caption{RMSE of univariate estimator averaging weights relative to oracle with $\rho=0.9$}
\label{tab:uni_weight_0.9}
\begin{tabular}{lccc}
\toprule
$T$ & $h=1$ & $h=3$ & $h=6$ \\
\midrule
\multicolumn{4}{c}{AR(1)}\\
\hline
200 & 0.1683 & 0.0947 & 0.0856 \\
400 & 0.1398 & 0.0911 & 0.0850 \\
800 & 0.1136 & 0.0713 & 0.0583 \\
\addlinespace
\multicolumn{4}{c}{local-to-AR(1) with $\alpha=0.5$}\\
\hline
200 & 0.0095 & 0.4604 & 0.1793 \\
400 & 0.0126 & 0.3023 & 0.1544 \\
800 & 0.0167 & 0.1983 & 0.1543 \\
\addlinespace
\multicolumn{4}{c}{ARMA(1,1) with $\alpha=0.5$}\\
\hline
200 & 0.0095 & 0.4604 & 0.1793 \\
400 & 0.0018 & 0.2643 & 0.3138 \\
800 & 0.0001 & 0.1135 & 0.4164 \\
\addlinespace
\multicolumn{4}{c}{local-to-AR(1) with $\alpha=0.9$}\\
\hline
200 & 0.0006 & 0.3158 & 0.4447 \\
400 & 0.0002 & 0.1724 & 0.4560 \\
800 & 0.0007 & 0.1335 & 0.3628 \\
\addlinespace
\multicolumn{4}{c}{ARMA(1,1) with $\alpha=0.9$}\\
\hline
200 & 0.0006 & 0.3158 & 0.4447 \\
400 & 0.0000 & 0.0979 & 0.6427 \\
800 & 0.0000 & 0.0000 & 0.4741 \\
\bottomrule
\end{tabular}
\end{table}

\begin{sidewaystable}[!htbp]
\centering
\caption{RMSE of univariate IRF estimators with $\rho=0.9$}
\label{tab:uni_rmse_0.9}
\resizebox{\textwidth}{!}{%
\begin{tabular}{lcccccccccccccccccc}
\toprule
& \multicolumn{6}{c}{$h=1$} & \multicolumn{6}{c}{$h=3$} & \multicolumn{6}{c}{$h=6$} \\
\cmidrule(lr){2-7}\cmidrule(lr){8-13}\cmidrule(lr){14-19}
$T$ & $\widehat{\theta}_{VAR}$ & $\widehat{\theta}_{LP}$ & $\widehat{\theta}_O$ & $\widehat{\theta}_P$ & $\widehat{\theta}_{TLP}$ & $\widehat{\theta}_M$ & $\widehat{\theta}_{VAR}$ & $\widehat{\theta}_{LP}$ & $\widehat{\theta}_O$ & $\widehat{\theta}_P$ & $\widehat{\theta}_{TLP}$ & $\widehat{\theta}_M$ & $\widehat{\theta}_{VAR}$ & $\widehat{\theta}_{LP}$ & $\widehat{\theta}_O$ & $\widehat{\theta}_P$ &$\widehat{\theta}_{TLP}$ & $\widehat{\theta}_M$ \\
\midrule
\multicolumn{19}{c}{AR(1)}\\
\hline
200 & 0.0420 & 0.0727 & 0.0417 & 0.0484 & 0.0530 & 0.0500 & 0.0941 & 0.1215 & 0.0940 & 0.0951 & 0.1029 & 0.0982 & 0.1236 & 0.1587 & 0.1236 & 0.1239 & 0.1367 & 0.1260 \\
400 & 0.0249 & 0.0500 & 0.0249 & 0.0285 & 0.0343 & 0.0327 & 0.0578 & 0.0822 & 0.0580 & 0.0583 & 0.0667 & 0.0622 & 0.0797 & 0.1023 & 0.0799 & 0.0800 & 0.0882 & 0.0814 \\
800 & 0.0168 & 0.0351 & 0.0168 & 0.0192 & 0.0237 & 0.0225 & 0.0400 & 0.0584 & 0.0400 & 0.0403 & 0.0468 & 0.0436 & 0.0570 & 0.0726 & 0.0570 & 0.0570 & 0.0622 & 0.0584 \\
\addlinespace
\multicolumn{19}{c}{local-to-AR(1), $\alpha=0.5$}\\
\hline
200 & 0.4684 & 0.1199 & 0.1199 & 0.1213 & 0.1360 & 0.2861 & 0.3277 & 0.2031 & 0.2033 & 0.2236 & 0.2478 & 0.2593 & 0.1880 & 0.2547 & 0.1868 & 0.1981 & 0.2074 & 0.1865 \\
400 & 0.3216 & 0.0589 & 0.0589 & 0.0607 & 0.0692 & 0.1808 & 0.2079 & 0.1183 & 0.1173 & 0.1272 & 0.1459 & 0.1541 & 0.1015 & 0.1501 & 0.1014 & 0.1062 & 0.1144 & 0.1037 \\
800 & 0.2216 & 0.0366 & 0.0366 & 0.0383 & 0.0428 & 0.1199 & 0.1334 & 0.0757 & 0.0731 & 0.0767 & 0.0889 & 0.0952 & 0.0592 & 0.0959 & 0.0592 & 0.0638 & 0.0707 & 0.0636 \\
\addlinespace
\multicolumn{19}{c}{ARMA(1,1), $\alpha=0.5$}\\
\hline
200 & 0.4684 & 0.1199 & 0.1199 & 0.1213 & 0.1360 & 0.2861 & 0.3277 & 0.2031 & 0.2033 & 0.2236 & 0.2478 & 0.2593 & 0.1880 & 0.2547 & 0.1868 & 0.1981 & 0.2074 & 0.1865 \\
400 & 0.4613 & 0.1050 & 0.1050 & 0.1051 & 0.1127 & 0.2779 & 0.3079 & 0.1482 & 0.1482 & 0.1594 & 0.1936 & 0.2298 & 0.1511 & 0.1742 & 0.1450 & 0.1482 & 0.1516 & 0.1447 \\
800 & 0.4581 & 0.0999 & 0.0999 & 0.0999 & 0.1038 & 0.2756 & 0.2989 & 0.1164 & 0.1164 & 0.1182 & 0.1502 & 0.2158 & 0.1335 & 0.1271 & 0.1179 & 0.1228 & 0.1246 & 0.1231 \\
\addlinespace
\multicolumn{19}{c}{local-to-AR(1), $\alpha=0.9$}\\
\hline
200 & 0.8613 & 0.4055 & 0.4055 & 0.4055 & 0.4171 & 0.6285 & 0.6313 & 0.3969 & 0.3969 & 0.4102 & 0.4828 & 0.5310 & 0.3826 & 0.3734 & 0.3597 & 0.3526 & 0.3693 & 0.3646 \\
400 & 0.5941 & 0.1818 & 0.1818 & 0.1818 & 0.1883 & 0.3838 & 0.4085 & 0.1977 & 0.1977 & 0.2014 & 0.2559 & 0.3145 & 0.2123 & 0.2058 & 0.1942 & 0.1911 & 0.2002 & 0.1990 \\
800 & 0.4099 & 0.0785 & 0.0785 & 0.0785 & 0.0828 & 0.2404 & 0.2633 & 0.1032 & 0.1032 & 0.1058 & 0.1337 & 0.1876 & 0.1134 & 0.1189 & 0.1039 & 0.1093 & 0.1099 & 0.1060 \\
\addlinespace
\multicolumn{19}{c}{ARMA(1,1), $\alpha=0.9$}\\
\hline
200 & 0.8613 & 0.4055 & 0.4055 & 0.4055 & 0.4171 & 0.6285 & 0.6313 & 0.3969 & 0.3969 & 0.4102 & 0.4828 & 0.5310 & 0.3826 & 0.3734 & 0.3597 & 0.3526 & 0.3693 & 0.3646 \\
400 & 0.8550 & 0.3960 & 0.3960 & 0.3960 & 0.4017 & 0.6217 & 0.6144 & 0.3537 & 0.3537 & 0.3545 & 0.4212 & 0.5059 & 0.3529 & 0.3037 & 0.3037 & 0.2904 & 0.3247 & 0.3322 \\
800 & 0.8521 & 0.3934 & 0.3934 & 0.3934 & 0.3961 & 0.6194 & 0.6065 & 0.3322 & 0.3322 & 0.3322 & 0.3739 & 0.4938 & 0.3388 & 0.2643 & 0.2643 & 0.2565 & 0.3002 & 0.3149 \\
\bottomrule
\end{tabular}
}
\begin{minipage}{\textwidth}\footnotesize
Note: $\widehat{\theta}_{VAR}$ and $\widehat{\theta}_{LP}$ are the VAR and LP IRF estimators.
$\widehat{\theta}_O$ uses the infeasible oracle min-MSE weight $w_h^\star$ in \eqref{eq:wstar}.
$\widehat{\theta}_P$ uses the sieve-bootstrap plug-in weight $\widehat w_P$ (Algorithm~\ref{alg:SV-weight}).
$\widehat{\theta}_{TLP}$ is the TLP estimator. $\widehat{\theta}_M$ uses the $R^2$-based model-averaging weight $\widehat w_{M}$ in \eqref{eq:wr2}. The autoregressive coefficient ($\rho$) is 0.9 in each case. Results are based on 1,000 Monte Carlo replications and 500 bootstrap iterations.
\end{minipage}
\end{sidewaystable}

\newpage
\subsection{Multivariate DGP Details}
\label{app:multi_details}
The coefficient matrices for the multivariate simulations in Section \ref{subsec:MC_multi} are defined as follows. For both specifications, the dimension is $n=3$, the autoregressive order is $p=4$, and the structural impact matrix is defined below.

\subsection*{SVAR(4) Specification}
The autoregressive coefficients $A_1^s, \dots, A_4^s$ are:
\begin{align*}
A_1^s &= \begin{pmatrix} 
1.31 & 0.75 & 0.25 \\ 
-0.12 & 2.08 & 0.23 \\ 
-0.23 & 0.56 & 1.75 
\end{pmatrix}, &
A_2^s &= \begin{pmatrix} 
-0.52 & -1.06 & -0.35 \\ 
0.16 & -1.59 & -0.33 \\ 
0.32 & -0.78 & -1.12 
\end{pmatrix}, \\[1em]
A_3^s &= \begin{pmatrix} 
0.04 & 0.48 & 0.16 \\ 
-0.08 & 0.53 & 0.15 \\ 
-0.14 & 0.35 & 0.31 
\end{pmatrix}, &
A_4^s &= \begin{pmatrix} 
0.01 & -0.07 & -0.02 \\ 
0.01 & -0.06 & -0.02 \\ 
0.02 & -0.05 & -0.03 
\end{pmatrix}.
\end{align*}
The structural impact matrix $M_0^s$ is:
\begin{equation*}
M_0^s = \begin{pmatrix} 
2.0 & -1.5 & 0.2 \\
1.7 & 1.3 & 0.7 \\
0.6 & -0.6 & 1.7 
\end{pmatrix}
\end{equation*}

\subsection*{SVARMA(4,1) Specification}
The autoregressive coefficients $A_1^m, \dots, A_4^m$ are:
\begin{align*}
A_1^m &= \begin{pmatrix} 
1.24 & -0.04 & -0.03 \\ 
-0.58 & 1.77 & 0.32 \\ 
-0.78 & 0.76 & 1.63 
\end{pmatrix}, &
A_2^m &= \begin{pmatrix} 
-0.52 & 0.02 & 0.06 \\ 
0.74 & -1.23 & -0.39 \\ 
1.04 & -0.98 & -1.02 
\end{pmatrix}, \\[1em]
A_3^m &= \begin{pmatrix} 
0.08 & 0.00 & -0.03 \\ 
-0.30 & 0.39 & 0.16 \\ 
-0.44 & 0.41 & 0.29 
\end{pmatrix}, &
A_4^m &= \begin{pmatrix} 
-0.01 & 0.00 & 0.00 \\ 
0.04 & -0.04 & -0.02 \\ 
0.06 & -0.05 & -0.03 
\end{pmatrix}.
\end{align*}
The moving average matrix $M_1^m$ is:
\begin{equation*}
M_1^m = \begin{pmatrix} 
-0.30 & 0.10 & -0.40 \\
-0.20 & 0.20 & -1.00 \\
-0.30 & 0.07 & 0.20 
\end{pmatrix}
\end{equation*}
The structural impact matrix $M_0^m$ is:
\begin{equation*}
M_0^m = \begin{pmatrix} 
1.30 & 0.40 & 0.10 \\
-0.02 & 0.05 & 2.00 \\
-0.08 & -1.70 & 0.80 
\end{pmatrix}
\end{equation*}

\newpage
\subsection{Additional Multivariate Simulation Results}
\label{app:DGP_matrices}

\begin{table}[!htbp]
\centering
\caption{RMSE of estimator averaging weights relative to oracle}
\label{tab:multi_weight_rmse}
\begin{tabular}{lccccccccc}
\toprule
& \multicolumn{3}{c}{SVAR(4)} 
& \multicolumn{3}{c}{local-to-SVAR(4)} 
& \multicolumn{3}{c}{SVARMA(4,1)} \\
\cmidrule(lr){2-4}\cmidrule(lr){5-7}\cmidrule(lr){8-10}
$T$ 
& $h=1$ & $h=6$ & $h=12$
& $h=1$ & $h=6$ & $h=12$
& $h=1$ & $h=6$ & $h=12$ \\
\midrule
200  
& 0.2690 & 0.2985 & 0.0188
& 0.2921 & 0.0981 & 0.0141
& 0.2921 & 0.0981 & 0.0141 \\
800  
& 0.1642 & 0.1777 & 0.0273
& 0.2928 & 0.1147 & 0.0160
& 0.2946 & 0.1132 & 0.0178 \\
2000 
& 0.0994 & 0.1440 & 0.0581
& 0.2787 & 0.1011 & 0.0445
& 0.2967 & 0.1797 & 0.0330 \\
\bottomrule
\end{tabular}
\end{table}

Figure \ref{fig:true_irf_plots} displays the true impulse response functions alongside the average estimates from the VAR and LP estimators ($T=200$).

\begin{figure}[!htbp]
  \centering
  \begin{subfigure}[t]{0.48\textwidth}
    \centering
    \includegraphics[width=\textwidth]{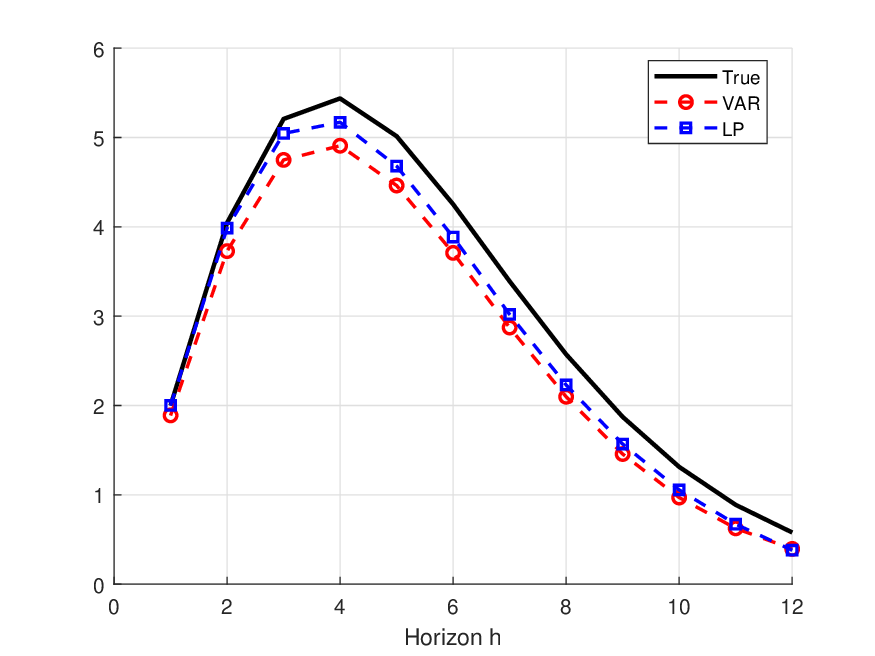}
    \caption{SVAR(4)}
    \label{fig:irfVAR4}
  \end{subfigure}\hfill
  \begin{subfigure}[t]{0.48\textwidth}
    \centering
    \includegraphics[width=\textwidth]{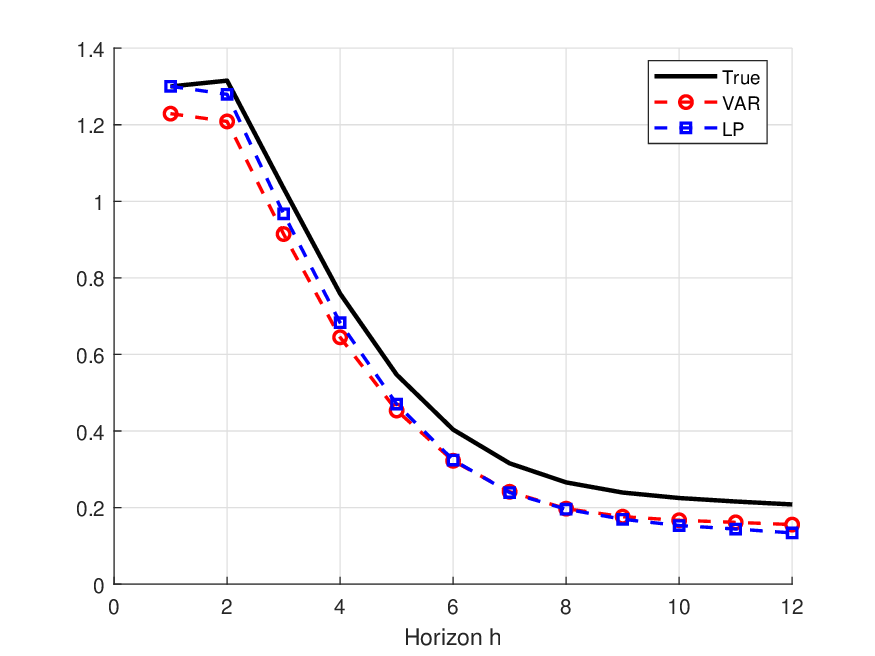}
    \caption{SVARMA(4,1)}
    \label{fig:irfVARMA41}
  \end{subfigure}
  \caption{True impulse responses and average VAR and LP estimates for two multivariate data-generating processes. Results are based on a sample size of $T=200$, using 1,000 Monte Carlo replications and 500 bootstrap iterations. VAR lag order is selected via AIC and is also applied to the LP estimations.}
  \label{fig:true_irf_plots}
\end{figure}

\newpage
\section{Additional Results for the Application}
\label{app:IRFs}

\begin{figure}[htbp]
  \centering
  \includegraphics[width=\textwidth]{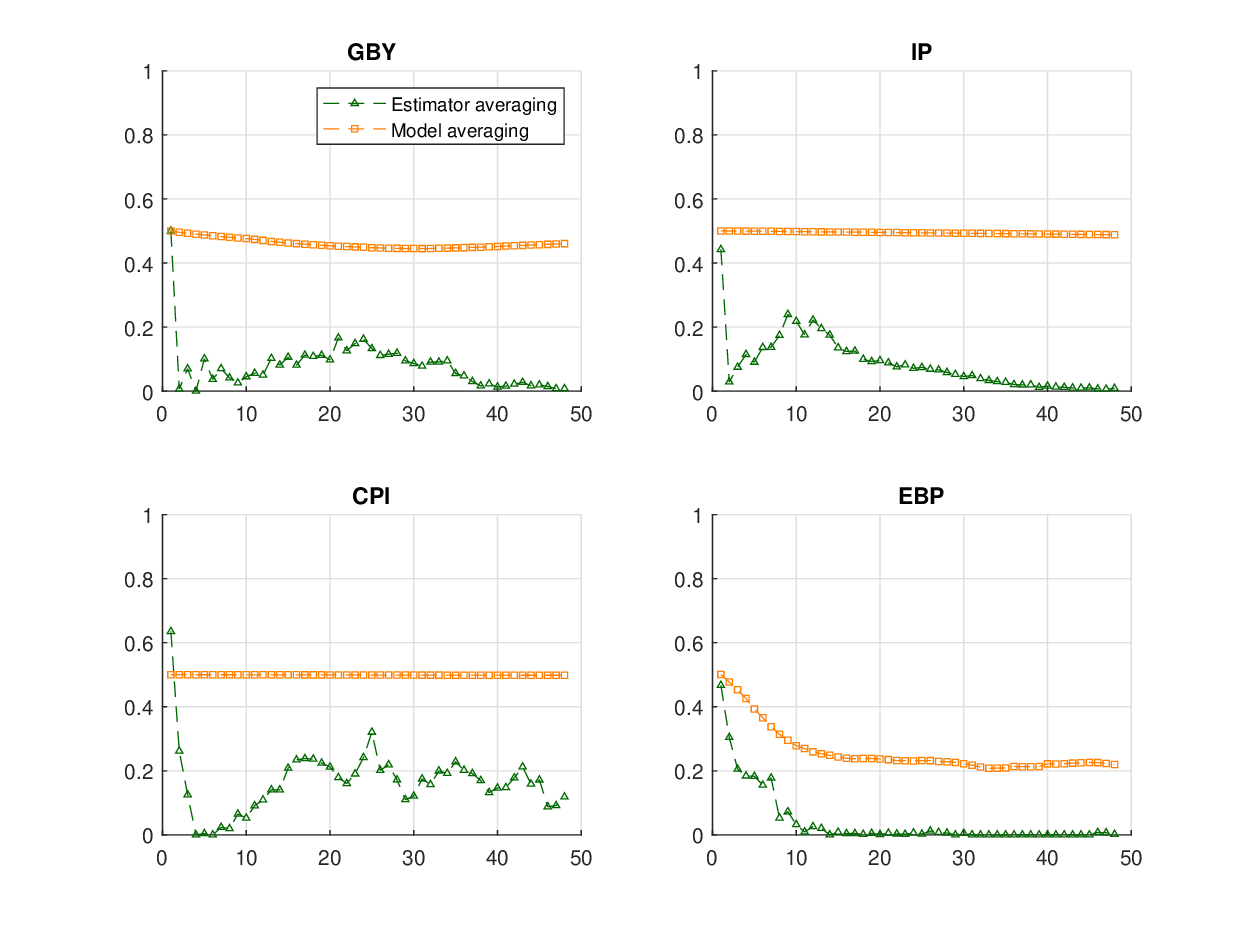}
  \caption{Estimated weights on the IV-LP estimator across the model-averaging and estimator-averaging approaches for each macroeconomic variable.}
  \label{fig:emp:weights}
\end{figure}

\begin{figure}[htbp]
  \centering
  \includegraphics[width=\textwidth]{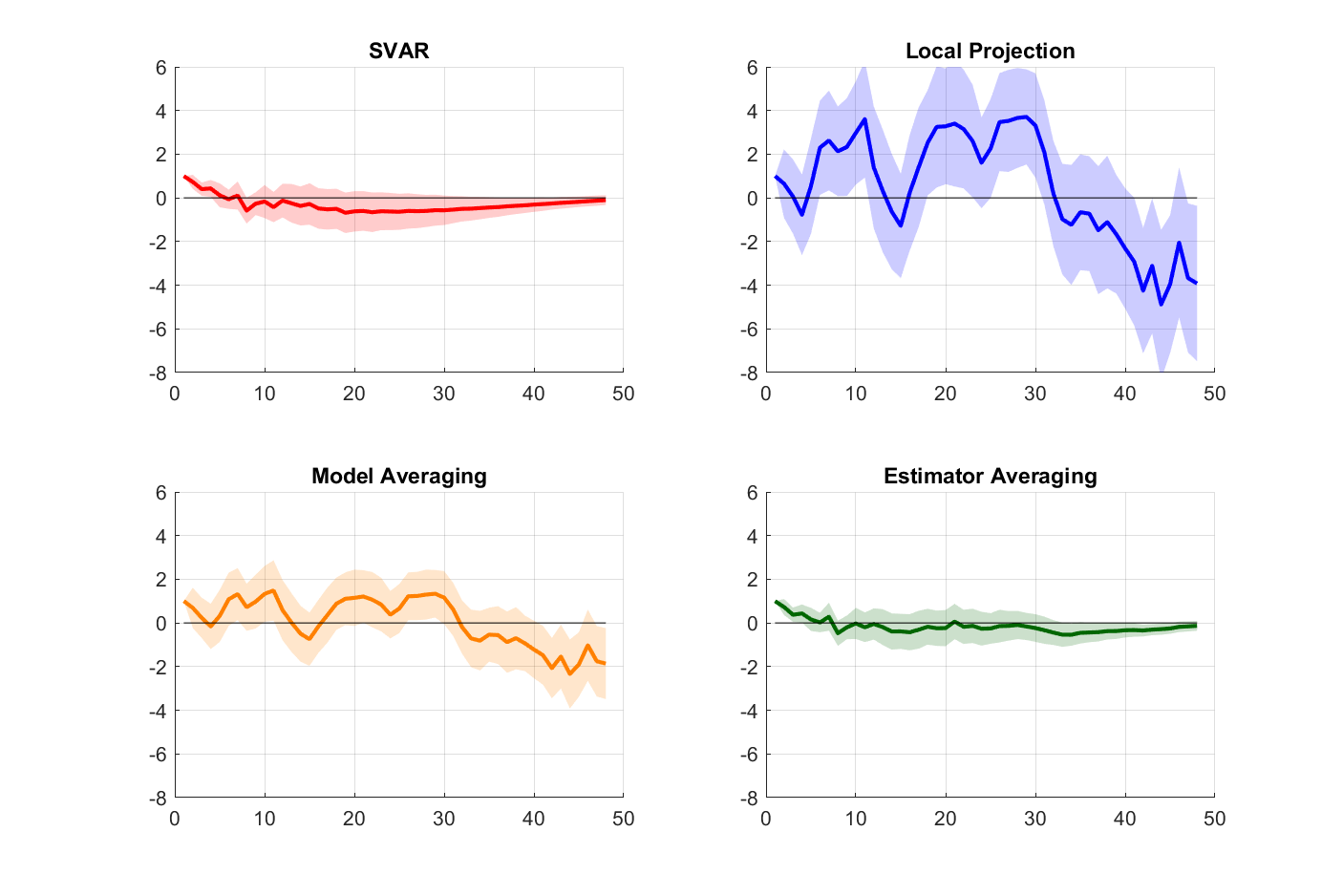}
  \caption{Estimated responses of the two-year Treasury yield to a 25-basis-point contractionary monetary policy shock. The four panels display the estimates from the IV-SVAR, IV-LP, model-averaging, and estimator-averaging approaches. The 68 percent pointwise confidence intervals are calculated using a nested VAR-sieve wild bootstrap procedure (500 inner and outer iterations) and are centered around the point estimates.}
  \label{fig:emp:GBY}
\end{figure}

\begin{figure}[htbp]
  \centering
  \includegraphics[width=\textwidth]{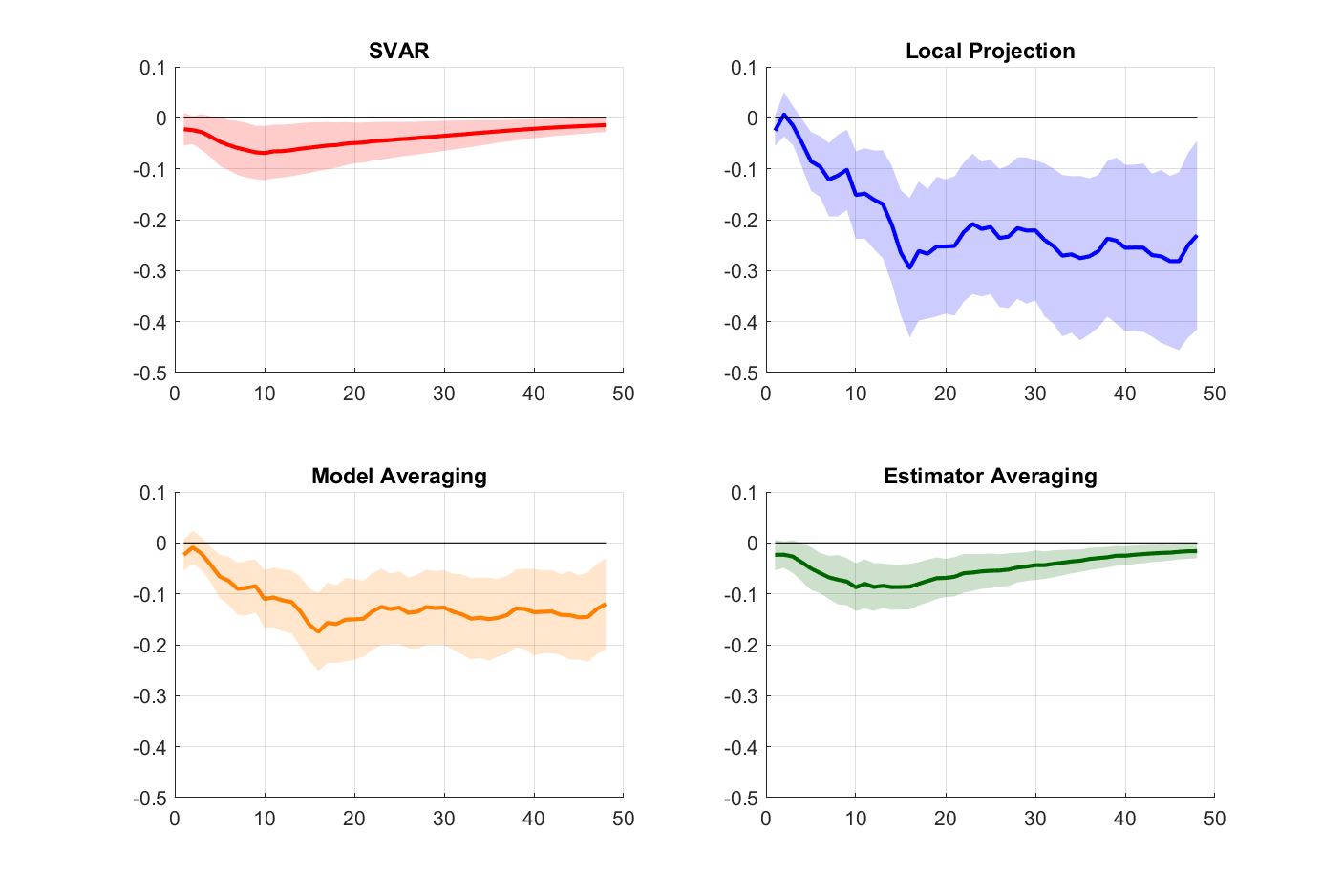}
  \caption{Estimated responses of industrial production to a 25-basis-point contractionary monetary policy shock. The four panels display the estimates from the IV-SVAR, IV-LP, model-averaging, and estimator-averaging approaches. The 68 percent pointwise confidence intervals are calculated using a nested VAR-sieve wild bootstrap procedure (500 inner and outer iterations) and are centered around the point estimates.}
  \label{fig:emp:IP}
\end{figure}

\begin{figure}[htbp]
  \centering
  \includegraphics[width=\textwidth]{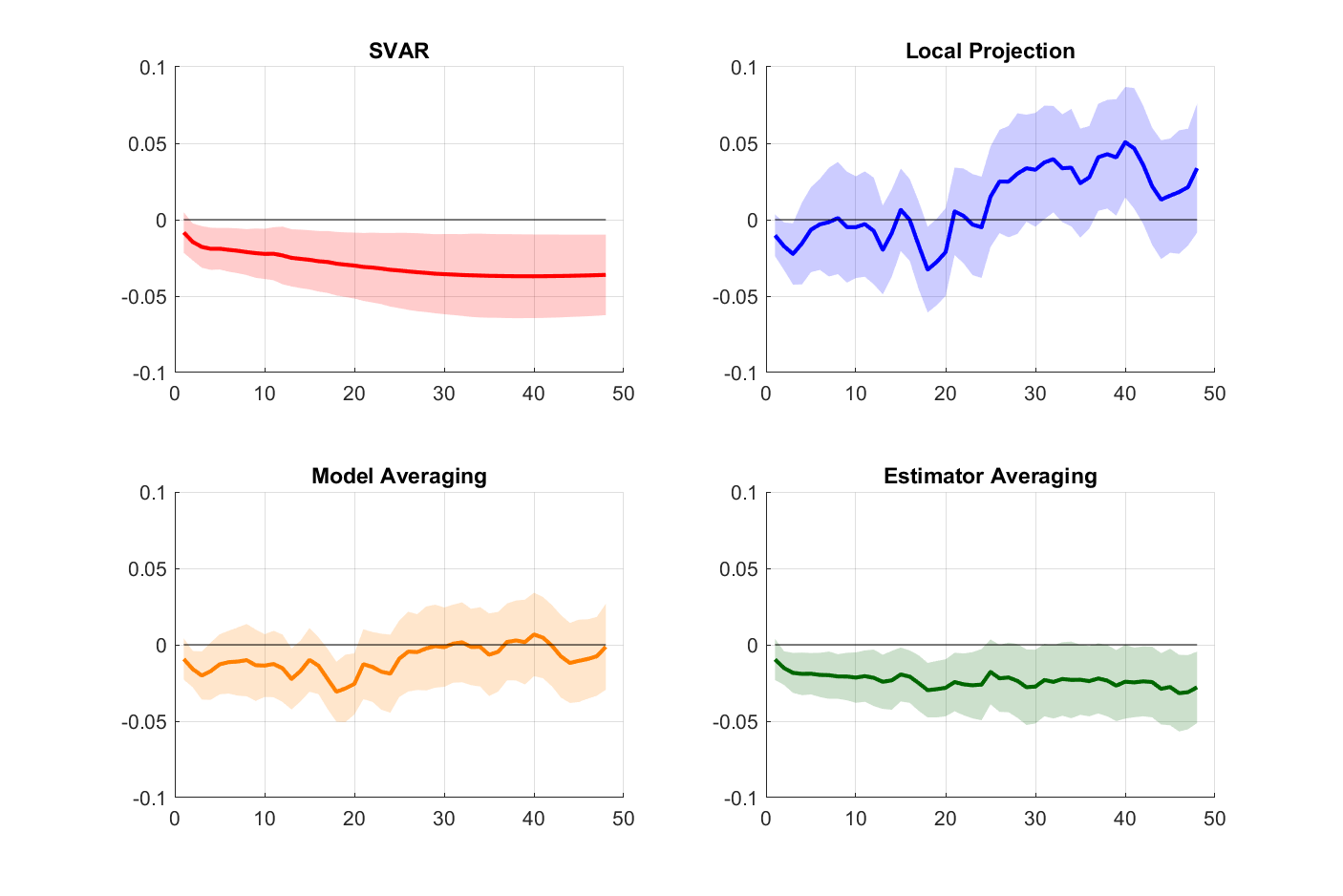}
  \caption{Estimated responses of consumer prices to a 25-basis-point contractionary monetary policy shock. The four panels display the estimates from the IV-SVAR, IV-LP, model-averaging, and estimator-averaging approaches. The 68 percent pointwise confidence intervals are calculated using a nested VAR-sieve wild bootstrap procedure (500 inner and outer iterations) and are centered around the point estimates.}
  \label{fig:emp:CPI}
\end{figure}

\begin{figure}[htbp]
  \centering
  \includegraphics[width=\textwidth]{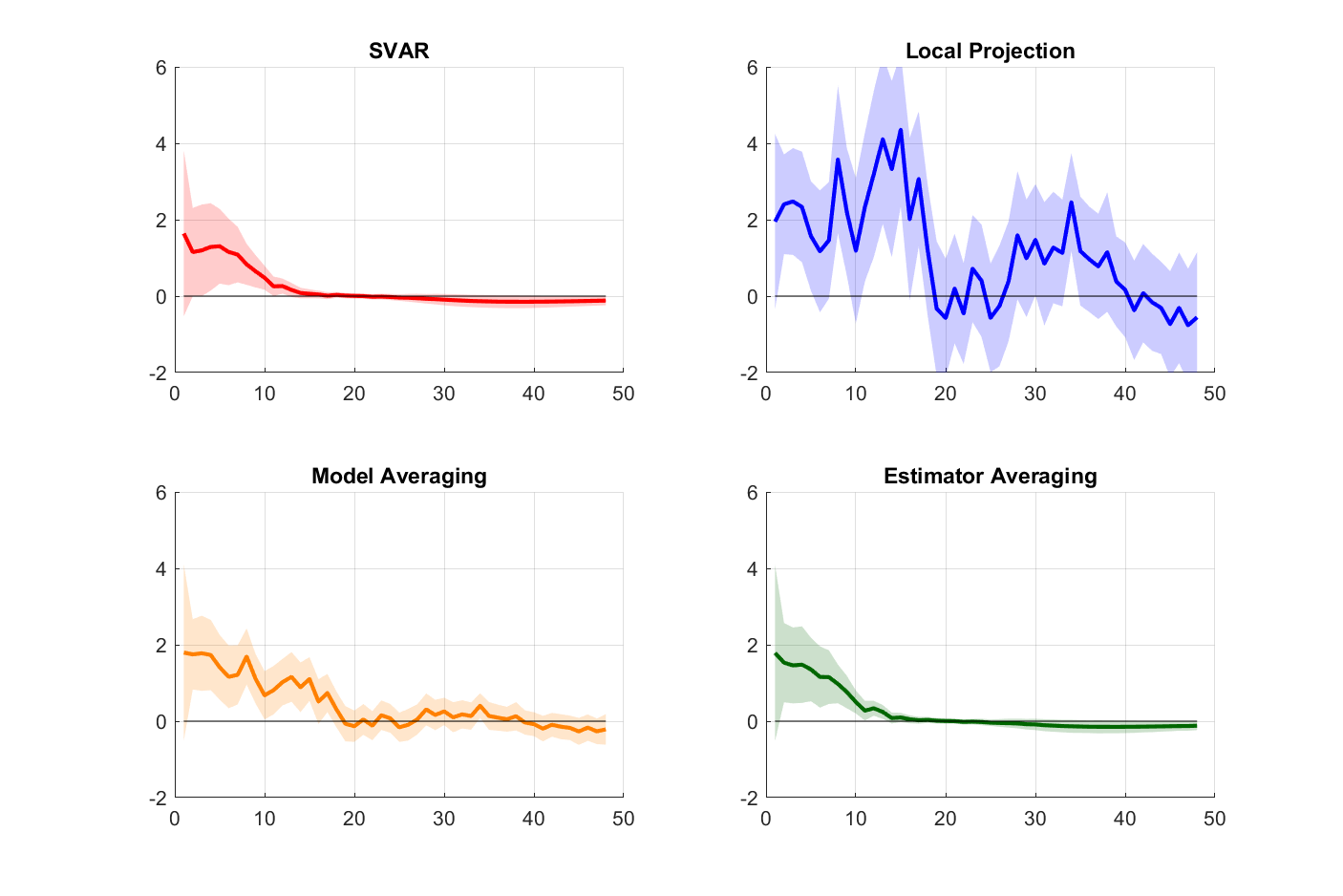}
  \caption{Estimated responses of the excess bond premium to a 25-basis-point contractionary monetary policy shock. The four panels display the estimates from the IV-SVAR, IV-LP, model-averaging, and estimator-averaging approaches. The 68 percent pointwise confidence intervals are calculated using a nested VAR-sieve wild bootstrap procedure (500 inner and outer iterations) and are centered around the point estimates.}
  \label{fig:emp:EBP}
\end{figure}

\newpage
\singlespacing

\end{appendices}
\end{document}